\documentclass[lettersize,journal]{IEEEtran}

\usepackage{moreverb,url}
\usepackage[colorlinks,bookmarksopen,bookmarksnumbered,citecolor=red,urlcolor=red]{hyperref}
\usepackage{hyperref}

\usepackage{amsmath} 
\usepackage{amssymb}  
\usepackage{amsthm}
\usepackage{amsfonts}
\usepackage{mathtools}
\usepackage{bm}
\providecommand{\bm}{\pmb}

\newtheorem{defin}{Definition}

\newtheorem{thm}{Theorem}

\theoremstyle{definition}

\theoremstyle{remark}
\newtheorem{rmk}{Remark}

\usepackage{verbatim}
\usepackage{color}
\usepackage{graphicx}
\usepackage{caption}
\usepackage{mathptmx}
\usepackage{times}
\usepackage{subfig}

\usepackage[us]{datetime}
\usepackage[usenames,dvipsnames,table]{xcolor}

\hypersetup{
    colorlinks,
    citecolor=black,
    filecolor=black,
    linkcolor=black,
    urlcolor=black,
    pdfauthor={},
    pdfsubject={},
    pdftitle={}
}

\usepackage{todonotes}

\graphicspath{{images/}}

\usepackage{graphicx}

\usepackage{latexsym}

\usepackage{tabularx, booktabs}
\newcolumntype{Y}{>{\centering\arraybackslash}X}
\usepackage{multirow}
\usepackage{paralist}

\usepackage{color}

\newcommand{\errorRL}{e_{\rotMatL}}
\newcommand{\vect}[1]{\bm{#1}}		
\newcommand{\matr}[1]{\bm{#1}}		
\newcommand{\nR}[1]{\mathbb{R}^{#1}}		
\newcommand{\SO}[1]{SO({#1})}		
\newcommand{\matrice}[1]{\begin{bmatrix} #1 \end{bmatrix}}	
\newcommand{\upperRomannumeral}[1]{\uppercase\expandafter{\romannumeral#1}}	

\newcommand{\vSpace}{\;\,}

\newcommand{\diag}[1]{\text{diag}\left( #1 \right)}

\newcommand{\fig}{Fig.~}	

\renewcommand{\frame}{\mathcal{F}}		
\newcommand{\origin}{O}						
\newcommand{\vX}{\vect{x}}					
\newcommand{\vY}{\vect{y}}					
\newcommand{\vZ}{\vect{z}}					
\newcommand{\vE}[1]{\vect{e}_{#1}}			
\newcommand{\vV}{\vect{v}}					
\newcommand{\pos}{\vect{p}}				
\newcommand{\dpos}{\dot{\vect{p}}}		
\newcommand{\ddpos}{\ddot{\vect{p}}}	
\newcommand{\rotMat}{\matr{R}}			
\newcommand{\rotMatVectAngle}[2]{\rotMat_{#1}(#2)}	
\newcommand{\angVel}{\vect{\omega}}				
\newcommand{\vZero}{\vect{0}}				
\renewcommand{\skew}[1]{\matr{S}(#1)}				
\newcommand{\eye}[1]{\matr{I}_{#1}}		
\newcommand{\roll}{\phi}		
\newcommand{\pitch}{\theta}		
\newcommand{\pitchDes}{\bar{\theta}}		
\newcommand{\pitchEq}{\theta^{eq}}		
\newcommand{\yaw}{\psi}		
\newcommand{\yawDes}{\bar{\psi}}		

\newcommand{\frameW}{\frame_W}			
\newcommand{\originW}{\origin_W}		
\newcommand{\xW}{\vX_W}				
\newcommand{\yW}{\vY_W}				
\newcommand{\zW}{\vZ_W}				

\newcommand{\frameL}{\frame_L}			
\newcommand{\originL}{\origin_L}			
\newcommand{\xL}{\vX_L}				
\newcommand{\yL}{\vY_L}				
\newcommand{\zL}{\vZ_L}				
\newcommand{\pL}{\pos_L}			
\newcommand{\pLW}{\prescript{W}{}{\pL}} 
\newcommand{\dpL}{\dpos_L}			
\newcommand{\ddpL}{\ddpos_L}		
\newcommand{\rotMatL}{\rotMat_L}
\newcommand{\rotMatLEquilib}{\rotMat^{eq}_{L}}

\newcommand{\angVelL}{\angVel_L}		
\newcommand{\massL}{{m_L}}
\newcommand{\massLU}{\hat{{m}}_L}				
\newcommand{\inertiaL}{\matr{J}_L}	
\newcommand{\InertiaL}{\matr{M}_L}	
\newcommand{\coriolisL}{\vect{c}_L}	
\newcommand{\gravityL}{\vect{g}_L}	
\newcommand{\graspL}{\matr{G}}		

\newcommand{\length}[1]{{l}_{0#1}}
\newcommand{\lengthU}[1]{\hat{{l}}_{0#1}}	
\newcommand{\springCoeff}[1]{{k}_{#1}}
\newcommand{\springCoeffU}[1]{\hat{{k}}_{#1}}
\newcommand{\cableForce}[1]{\vect{f}_{#1}}
\newcommand{\cableForceEquilib}[1]{\vect{f}^{eq}_{#1}}

\newcommand{\cableAttitude}[1]{\vect{l}_{#1}}

\newcommand{\condZero}{\xi}
\newcommand{\anchorPoint}[1]{B_{#1}}			
\newcommand{\anchorPos}[1]{\vect{b}_{#1}}		
\newcommand{\anchorLength}[1]{{b}_{#1}}	
\newcommand{\anchorLengthU}[1]{\hat{{b}}_{#1}}		
\newcommand{\anchorPosL}[1]{\prescript{L}{}{\vect{b}}_{#1}}		
\newcommand{\anchorPosLU}[1]{\prescript{L}{}{\hat{\vect{b}}}_{#1}}		

\newcommand{\cableForces}{\cableForce{}}		

\newcommand{\frameR}[1]{\frame_{R #1}}			
\newcommand{\originR}[1]{O_{R #1}}					
\newcommand{\xR}[1]{\vX_{R #1}}								
\newcommand{\yR}[1]{\vY_{R #1}}								
\newcommand{\zR}[1]{\vZ_{R #1}}								
\newcommand{\pR}[1]{\pos_{R #1}}						
\newcommand{\dpR}[1]{\dpos_{R #1}}					
\newcommand{\ddpR}[1]{\ddpos_{R #1}}				
\newcommand{\uR}[1]{\vect{u}_{R #1}}				
\newcommand{\rotMatR}[1]{\rotMat_{R #1}}			

\newcommand{\dampingA}[1]{\matr{B}_{A#1}}		
\newcommand{\springA}[1]{\matr{K}_{A#1}}		
\newcommand{\inertiaA}[1]{\matr{M}_{A#1}}		
\newcommand{\paramA}[1]{\vect{\pi}_{A#1}}			

\newcommand{\config}{\vect{q}}					
\newcommand{\dconfig}{\vect{v}}			
\newcommand{\ddconfig}{\dot{\dconfig}}			
\newcommand{\configR}{\config_R}					
\newcommand{\dconfigR}{\dconfig_R}				
\newcommand{\configL}{\config_L}					
\newcommand{\dconfigL}{\dconfig_L}				
\newcommand{\ddconfigL}{\ddconfig_L}			
\newcommand{\state}{\vect{x}}						
\newcommand{\dynamicModelFun}{m}					

\newcommand{\configEq}{\bar{\config}}			
\newcommand{\configLEq}{\bar{\config}_L}		
\newcommand{\configREq}{\bar{\config}_R}		
\newcommand{\paramAEq}[1]{\bar{\vect{\pi}}_{A#1}}			
\newcommand{\paramAEqInc}[1]{\hat{\bar{\vect{\pi}}}_{A#1}}
\newcommand{\paramAEqU}[1]{\hat{\bar{\vect{\pi}}}_{A#1}}
\newcommand{\pLEq}{\bar{\pos}_L}

\newcommand{\pLEquilib}{{\pos}^{eq}_L}
\newcommand{\rotMatLEq}{\bar{\rotMat}_L}		
\newcommand{\cableForceEq}[1]{\bar{\vect{f}}_{#1}}
\newcommand{\cableForceEqInc}[1]{\hat{\bar{\vect{f}}}_{#1}}				
\newcommand{\cableForcesEq}{\bar{\vect{f}}}	
\newcommand{\cableForcesEqInc}{\hat{\bar{\vect{f}}}}					
\newcommand{\internalTension}{t_L}				
\newcommand{\pREq}[1]{\bar{\pos}_{R #1}}	
\newcommand{\pREqInc}[1]{\hat{\bar{\pos}}_{R #1}}		

\newcommand{\pREquilib}[1]{\pos^{eq}_{R #1}}

\newcommand{\configSetEq}{\mathcal{Q}(\internalTension,\configLEq)}
\newcommand{\configSetEqZero}{\mathcal{Q}(0,\configLEq)}
\newcommand{\configSetEqZeroi}[1]{\mathcal{Q}_{#1}(0,\configLEq)}

\newcommand{\configSetEqPlus}{\mathcal{Q}^+(\internalTension,\configLEq)}
\newcommand{\configSetEqMinus}{\mathcal{Q}^-(\internalTension,\configLEq)}

\newcommand{\errorpREq}[1]{\vect{e}_{R#1}}	
\newcommand{\errorPL}{{\vect{e}_{p}}_L}

\newcommand{\dVZeroSet}{\mathcal{E}}				
\newcommand{\stateSetEq}{\mathcal{X}(\internalTension,\configLEq)}

\newcommand{\stateSetEqZeroi}[1]{\mathcal{X}_{#1}(0,\configLEq)}
\newcommand{\stateSetEqPlus}{\mathcal{X}^+(\internalTension,\configLEq)}
\newcommand{\stateSetEqMinus}{\mathcal{X}^-(\internalTension,\configLEq)}

\newcommand{\lyapunovFun}{V(\state)}				

\newcommand{\Vadd}{V_R(\state)}				

\newcommand{\dlyapunovFun}{\dot{V}(\state)}				

\newcommand{\displacement}{\vect{d}}
\newcommand{\pREqIncRef}[1]{\pREqInc{#1}^r}
\newcommand{\rotMatRRef}[1]{\rotMatR{#1}^r}

\DeclarePairedDelimiter{\norm}{\lVert}{\rVert}

\usepackage{amsmath} 
\usepackage{amssymb}  
\usepackage{amsthm}
\usepackage{amsfonts}
\usepackage{mathtools}
\usepackage{flushend}

\usepackage{todonotes}

\usepackage{graphicx}
\usepackage{color}

\usepackage{caption}

\begin{document}

\title{Equilibria, Stability, and Sensitivity for the Aerial Suspended Beam Robotic System subject to Parameter Uncertainty}

\author{C. Gabellieri$^1$~\IEEEmembership{IEEE~Member}, M. Tognon$^2$~\IEEEmembership{IEEE~Member}, D. Sanalitro$^3$~\IEEEmembership{IEEE~Member}, \\A. Franchi$^{1, 3,4}$,~\IEEEmembership{IEEE~Fellow}
\thanks{This work has been partially funded by the European Union's Horizon Europe research and innovation program [grant agreement No. ID: 101059875] Flyflic and  Horizon 2020 research and innovation programme [grant agreement No. 871479] AERIAL-CORE.}
\thanks{$^1$Robotics and Mechatronics group, EEMCS faculty, University of Twente, Enschede, The Netherlands $^2$Inria, Univ Rennes, CNRS, IRISA, Campus de Beaulieu, 35042 Rennes Cedex, France $^3$LAAS-CNRS, Universit\'e de Toulouse, CNRS, Toulouse, France.
       Email: {\tt\footnotesize c.gabellieri@utwente.nl, 	 a.franchi@utwente.nl} $^4$Department of Computer, Control and Management Engineering, Sapienza University of Rome, 00185 Rome, Italy,}
       
}   
\markboth{Transactions on Robotics,
}
{C. Gabellieri \MakeLowercase{\textit{(et al.)}:
Equilibria, Stability, and Sensitivity for the Aerial Suspended Beam Robotic System subject to
Parameter Uncertainty}}
\maketitle

\begin{abstract}
This work studies how  parametric uncertainties affect the cooperative manipulation of a cable-suspended beam-shaped load by means of two aerial robots not explicitly communicating with each other.
In particular, the work sheds light on the impact of the uncertain knowledge of the model parameters available to an established communication-less force-based controller. First, we find the closed-loop equilibrium configurations in the presence of the aforementioned uncertainties, and  then we study their stability. Hence, we show the fundamental role played in the robustness of the load attitude control by the internal force induced in the manipulated object by non-vertical cables. Furthermore, we formally study the sensitivity of the attitude error to such parametric variations, and we provide a method to act on the load position error in the presence of the uncertainties. Eventually, we validate the results through an extensive set of numerical tests in a realistic simulation environment including underactuated aerial vehicles and sagging-prone cables, and through hardware experiments.
\end{abstract}
\section{Introduction}\label{sec:intro}
\IEEEPARstart{I}{t} is nowadays universally acknowledged that the interest in Unmanned Aerial Vehicles (UAVs) is becoming  wider and wider by virtue of their ability to embrace an ample set of applications. 
A very recent and popular topic in aerial robotics is physical interaction using aerial manipulators~\cite{ollero2021past, 2018-RuLiOl,2018-KaJaAb} for applications such as contact-based inspection, assembly, human assistance, etc.
To solve these challenges, aerial platforms are endowed with physical interaction tools, such as cables~\cite{tognon2020theory} or more complex robotic arms~\cite{2017g-TogYueBuoFra}. 

Researchers have considered taking advantage of the cooperation between multiple robots to enhance the overall payload and manipulate large objects~\cite{skorobogatov2020multiple, 2010-MazKonBerOll, 2020-MohZweGan, loianno2017cooperative}. %
Different methods have been developed to tackle multi-robot aerial manipulation. In~\cite{2015-NguParLee} and~\cite{2013-RitD'an} the authors use passive manipulation tools to solve the cooperative aerial transportation of rigid and elastic objects, respectively. 
Multiple flying arms are instead used in~\cite{2015-CacGigMusPie, 2018-ThaBaiAco}. 
Cables have been often considered in multi-robot manipulation scenarios, because, in addition to being lightweight ad low-cost, they also mitigate the coupling between the system dynamics and the robots' attitude, which can simplify the control problem, especially when using underactuated aerial platforms.

\subsection{Related Works} The problem of manipulating a cable-suspended load through a team of aerial vehicles has been studied, e.g., in ~\cite{2013-SreKum,2016-MasBulSte,2013-ManDevRosCor, 2020-MoZwTaGa,2017-Lee, li2021cooperative}. In ~\cite{sanalitro2020full},  a robust pose controller for a cable-suspended load manipulated by multiple UAVs is presented. Stability is ensured through gain tuning given a bound on the uncertainties affecting the kinematic parameters. 
Formation control to transport the payload with a focus on robustness is described in~\cite{2020-RosRosGimSalSorSarCar}, where  modeling uncertainties are also taken into account.

A standard system that has attracted substantial interest in the research community is composed of \textit{two} aerial vehicles manipulating a \textit{beam-like load} through cables~\cite{2019-SpuPetVonSas, 2020-PeDi, sundin2022decentralized, 2020-MoZwTaGa,2017-TagKamVerSieNie, 2017-GasCieSca,2021-VilBraCarSar}.  Such standard configuration is of interest for several real-world applications, especially in the construction field, where we find columns, wooden pillars, iron beams for cement walls, scaffolds, pipes, pieces of roofs, and other beam-like building elements.  Two is the minimum number of aerial robots allowing to control both the position and attitude of a cable-suspended beam-like load~\cite{2020-PeDi}. While three aerial robots allow controlling the entire pose of a generic rigid body~\cite{2009-FinMicKimKum}, using more than two robots for a beam-like load it is arguably not the optimal solution in most of the cases because of the increased complexity of the system without being necessary for the control of the load.

In~\cite{2020-PeDi} and \cite{goodman2022geometric}, the authors propose a method for the transportation of a cable-suspended beam load by two aerial vehicles that have access to the state of the load; they consider rigid and elastic cables, respectively.  
In~\cite{sundin2022decentralized},  centralized and  decentralized model predictive control is proposed for a system of two UAVs manipulating a beam load through cables.

Decentralized algorithms as~\cite{2013-MelShoMicKum} are more robust and scalable with respect to (w.r.t.) the number of robots. 
However, decentralized \textit{communication-less} approaches have been also intensively studied in the literature \cite{wang2016force} because communication delays and packet losses are among the principal causes undermining the performance and stability of the system in real implementations, and because the hardware and software complexity can be reduced by confining explicit communication.  In~\cite{2017-GasCieSca}, a method relying on visual feedback is presented. As an alternative to vision, a force-based method  that uses admittance controllers and a leader-follower scheme is typically used to address communication-less aerial manipulation of cable-suspended objects~\cite{2017-TagKamVerSieNie, tagliabue2019robust, 2018h-TogGabPalFra,gabellieri2020study}. The leader robot guides the system following a predetermined trajectory, while the second robot, which carries a portion of the load weight, follows its lead by sensing the cable force variations.

 A primary goal of~\cite{2017-TagKamVerSieNie} is to keep the cables always vertical during transportation, meaning that no internal force  is induced in the object. The authors in~\cite{tagliabue2019robust} extend the results of~\cite{2017-TagKamVerSieNie} towards the $N-$robot case  and provides a method for tuning the gains of the robot admittance controllers in order to improve the robustness against disturbances induced by unmodeled dynamics and parametric uncertainties. In both works experiments are shown  in which, however, the altitude of the robots is set to a predetermined reference, implying either a centralized vertical movement coordination or restricted vertical motion.  

For such a popular class of communication-less, admittance-controlled, and leader-follower schemes, the formal analysis of the closed-loop system equilibrium configurations and their stability was presented for the first time in our previous work~\cite{2018h-TogGabPalFra}. 
There, we showed  that inducing an internal force on the load through non-vertical cables is required for full-pose regulation, especially to prevent  arbitrary vertical movements of the robots that would interfere with  the regulation of the load pitch and center position. In~\cite{gabellieri2020study}, we considered $N$ robots, empirically showing  through extensive simulations the effect of changing the number of leader robots on the stability and robustness against disturbances. Both works tackle only the ideal case where perfect knowledge of the system parameters is available to the admittance controllers of each aerial robot.

Despite being of primary interest, it has been unclear until now if and how in-practice-unavoidable uncertainties impact the pose regulation in the aforementioned control framework. Such a gap is filled in this work by introducing uncertainties on those system parameters used in the control action. 

Adaptive control laws have been proposed in the literature for the system in question, however, they are based on different assumptions than those used in this work.
For instance, the full state has been considered available for feedback in \cite{wang2016force}, or the robots rigidly attached to the object \cite{lee2016planning,lee2018integrated}. Other works assume the load mass is the sole uncertain parameter and only focus on the translational velocity regulation  \cite{2018-ThaBaiAco}, or rely on a communication network \cite{thapa2020cooperative}.  

In this work, we have found that the internal force induced by non-vertical cables plays a fundamental role in enabling task execution, especially in realistic conditions characterized by uncertainties. The importance of this is  masked when vertical movements of the robots are prevented or anyway the leader-follower approach is used solely to regulate the load motion on the horizontal plane, as in ~\cite{2017-TagKamVerSieNie, tagliabue2019robust}. On the other hand, the role of the internal force is crucial if the admittance-based  communication-less approach is applied in the full 3D space and, hence,   communication-less  full-pose regulation is sought.

While some loads may be damaged by internal forces, this is easily prevented in practice by enclosing the loads in suitable cases. Also, internal forces require additional control effort, which is justified by the benefits in terms of convergence and robustness of the load  pose control, as will be clear in the following. Indeed, internal forces have been often proposed also in the robotic grasping literature as a tool to make the grip on the object robust thanks to friction \cite{bicchi1994problem}.

\subsection{Contributions of the Work and Outline of the Paper}

The contribution of this work is showing the effects of  parametric uncertainties on the static regulation of the load pose when the usual approach~\cite{2017-TagKamVerSieNie, tagliabue2019robust, 2018h-TogGabPalFra,gabellieri2020study} based on admittance-controlled leader-follower aerial robots is used for  manipulating a cable-suspended beam in the absence of explicit communication.
In this approach, each robot knows only its own state and the force in its cable, retrievable from the robot's state using an external force observer. 

Note that, unlike in \cite{goodman2022geometric}, it is not feasible to assume all robots have knowledge of the object state. This is because the object state is based on the state of all robots, but data exchange among them is not considered in our scenario.

Throughout the manuscript,  we show that it is best to avoid the intuitive idea of having the cables vertical when performing the manipulation with force-based methods in real scenarios, i.e., when uncertainties are present. 
We address the problem aiming at a mathematically-sound point of view. We point out that we restrict the analysis to pose \textit{regulation}, hence, to quasi-static motion of the load. Tracking of more aggressive trajectories is left to future work.
With the above context in mind, the key contributions can be succinctly summarized:
\begin{itemize}
\item after formally studying the equilibrium configurations of the closed-loop system in the presence of uncertain parameters, their stability is proved using Lyapunov's  theory;
\item  the impact of an internal force induced by non-vertical cables on the \textit{robustness} of the load pose control is formally studied; 
\item  the effect of the internal force of diminishing the \textit{sensitivity} of the load attitude error to parametric uncertainty or \textit{variations} is shown; 
\item  a method for correcting the load position inaccuracy induced by the uncertainties is also presented; 
\item we present extensive numerical results and hardware experiments supporting the claims conveyed by the theoretical analysis;
\item last but not least, this work generalizes the system model by considering a generic position of the center of mass (CoM) of the load rather than assuming it to be exactly centered in the middle of the two anchoring points of the cables as done in~\cite{2018h-TogGabPalFra}.
\end{itemize}
 
 The paper is organized as follows.  Sec~\ref{sec:back} contains some background useful to better understand the results of the work. In  Secs~\ref{sec:equilibria} to  \ref{sec:error}, we present the three main contributions of the work: Sec~\ref{sec:equilibria} contains the derivation of the equilibrium points, and Sec~\ref{sec:stability} their stability analysis; Sec~\ref{sec:error} highlights the role of the internal forces in the load error robustness and sensitivity to parametric  variations. The results of the simulations and 
 experiments  are presented in Sec~\ref{sec:exp_num} and Sec~\ref{sec:exp_real}, respectively. Conclusive discussions are in Sec~\ref{sec:conclusions}.
\section{Background}\label{sec:back}
 %
 
 \begin{table}
\caption{Notation\textemdash General symbols and reference frames}\label{Tab:notation1}
\begin{tabularx}{\columnwidth}{@{}XX@{}}
\toprule
 $\eye{i}$ & $i\times i$ identity matrix\\
  $\vE{i}$ &  i$-{th}$ column of $\eye{3}$\\
  $\skew{\star}$ & skew operator\\
  $Ker(\star)$ & nullspace of $\star$\\
  diag($\star$) & diagonal matrix \\
  $ \star ^\top$ & transpose of $\star$\\
  $\norm{\star}$ & 2-norm of $\star$\\
  $\bar{\star}$ & desired value of $\star$\\
  $\star^{Eq}$ & value of $\star$ at the equilibrium\\
  $\hat{\star}$ & uncertain value of $\star$\\
  $\Delta_\star$ & $\star-\hat{\star}$\\
  $\frameW$ & inertial reference frame \\
  $\frameL$ & load reference frame \\
  $\frameR{i}$ & i$-{th}$ robot reference frame \\
  $\{\origin_A, \vX_A, \vY_A, \vZ_A\}$ & Origin, X$-$, Y$-$, and Z$-$ axis  of $\frame_A$ \\
  $^L\star$ & $\star$ expressed in $\frameL$\\
  $\dot{\star}$ & time derivative of $\star$\\
    \bottomrule
\end{tabularx}
\end{table}
\begin{table}
\caption{Notation\textemdash  System variables}\label{Tab:notation2}
\begin{tabularx}{\columnwidth}{@{}XX@{}}
\toprule
  $g$ & gravity acceleration \\
  $\pL, \dpL, \ddpL$ & load position, velocity, acceleration\\
  $\rotMatL$ & rotation of $\frameL$ w.r.t. $\frameW$\\
  $\angVelL$ & load angular velocity\\
  $\configL$ & set composed of $(\pL, \rotMatL)$\\
  $\dconfigL$ & $[\dpL^\top {^L\angVelL}^\top]^\top$\\
  $\yaw, \pitch$ & yaw and pitch of the load\\
  $\anchorPoint{i}$ & i$-{th}$ cable attaching point on load\\
  $\anchorPos{1}$ & position of point $\anchorPoint{i}$\\
  $\massL$, $\inertiaL$ & load mass and rotational inertia\\
  $L$ & load length, with $\ell=1/L$\\
  $\InertiaL$ & diag($\massL\eye{3}, \inertiaL)$\\
  $\gravityL$ & gravity terms in load dynamics\\
  $\coriolisL$ & Coriolis terms in load dynamics\\
  $\graspL$ & load grasp matrix\\
  $\pR{i}, \dpR{i}, \ddpR{i}$ & i$-{th}$ robot position, velocity, acceleration\\
  $\rotMatR{i}$ &rotation of $\frameR{i}$ w.r.t. $\frameW$\\
  $\configR$ & $[\pR{1}^\top\ \pR{2}^\top]^\top$\\
  $\springCoeff{i}$ & stiffness of the i$-th$ cable\\
  $\length{i}$ & rest length of the i$-th$ cable\\
  $\cableAttitude{i}$ & $\pR{i}-\anchorPos{i}$, vector along the i$-th$ cable\\
  $\cableForce{i}$ & force of the i$-th$ cable on the load\\
  $\vect{f}$ & $[\cableForce{1}^\top\ \cableForce{2}^\top]^\top$\\
  $\config$ & set $(\configR, \configR)$\\
  $\dconfig$ & $[\dconfigR^\top\ \dconfigL^\top]^\top$\\
  $\state$ & $(\config,\dconfig)$\\
  $\uR{i}$ & control input of robot i$-th$\\
  $\inertiaA{i}$ & i$-th$ robot control gain: apparent inertia, $\inertiaA{}$ diag($\inertiaA{1}, \inertiaA{2}$)\\
  $\dampingA{i}$ & i$-th$ robot control gain: apparent damping, $\dampingA{}$ diag($\dampingA{1}, \dampingA{2}$)\\
  $\springA{i}$ & i$-th$ robot control gain: virtual spring stiffness,  $\springA{}$ diag($\springA{1}, \springA{2}$)\\
  $\paramA{i}$ & feedforward term of i$-th$ robot's control, $\paramA=[\paramA{1}^\top\  \paramA{2}^\top]^\top$\\
  $\internalTension$ & load internal force\\
  $\condZero$ & $\left(\anchorLength{1}\massL - \frac{{\anchorLengthU{1}}{\massLU}L}{\hat{L}}\right)$\\
  $\dynamicModelFun(\config,\dconfig,\paramA{})$ & closed-loop dynamics\\
  $\configSetEq$ & $\{ \config \text{ satisfying Theorem~\ref{theorem:paramToConfig}} \}$\\
  ${\configSetEqPlus}$ & $\configSetEq$ where $\rotMatLEquilib\vE{1}$ has same sign of \eqref{eq:vec2}\\
   ${\rotMatLEquilib}^+$ & $\rotMatLEquilib$ in $\config\in{\configSetEqPlus}$\\
   ${\rotMatLEquilib}^-$ & ${\rotMatLEquilib}^+\rotMatVectAngle{\zL}{\pi}$\\
   ${\configSetEqMinus}$ & $\config\in\configSetEq$ s.t. $\rotMatLEquilib={\rotMatLEquilib}^-$\\
   $\configSetEqZeroi{1}$ & $\config\in\configSetEqZero$ s.t.  $(\rotMatLEquilib\vE{1})^\top\vE{3}=+1$\\
   $\configSetEqZeroi{2}$ & $\config\in\configSetEqZero$ s.t.  $(\rotMatLEquilib\vE{1})^\top\vE{3}=-1$\\
  $\stateSetEqZeroi{}$ & $\{ \state \; : \; \config \in \configSetEqZeroi{}, \; \dconfig = \vZero \}$\\
  $\stateSetEqZeroi{i}$&$\{ \state \; : \; \config \in \configSetEqZeroi{i}, \; \dconfig = \vZero \}$\\
 $\stateSetEq^\pm $ & $\{ \state \; : \; \config \in \configSetEq^\pm, \; \dconfig = \vZero \}$\\
 $\errorRL$, $\errorPL$ & load attitude and position errors\\
  \bottomrule
\end{tabularx}
\end{table}
In this section, we provide the background needed to understand the contribution of the work. Specifically, we quickly recall the system's main variables, the dynamics equations, and the already established findings.For the sake of readability, the notation is also summarized in Tables \ref{Tab:notation1} and \ref{Tab:notation2}.

The considered system, schematically shown in \fig\ref{fig:model:systemDescription}, is the typical rigid beam-like load attached to two aerial vehicles by means of  cables. 
\subsubsection{Load Model}
The beam-like load has mass $\massL \in \nR{}_{>0}$ and  positive-definite rotational inertia  ${\inertiaL }\in \nR{3 \times 3}$.
 The frame $\frameL = \{\originL,\xL,\yL,\zL\}$, where $\originL$ coincides with the load CoM, is rigidly attached to the load. 
The inertial frame is denoted with $\frameW = \{\originW,\xW,\yW,\zW\}$ where $\zW$ is oriented in the direction opposite to the gravity. 
The position and orientation of $\frameL$ w.r.t. $\frameW$, defined by the vector\footnote{The left superscript indicates the reference frame. From now on, $\frameW$ is considered as a reference frame when the superscript is omitted.}
 $\pLW \in \nR{3}$ and the rotation matrix $\rotMatL$, respectively, describe the full configuration of the load. We recall that the rotation along the axis that passes between the two cable anchoring points is not controllable by the robots. Being the load a beam, only the yaw angle, $\psi$, and pitch angle, $\theta$, are used to describe its attitude.
The usual equations of a rigid body subject to gravity and contact forces describe the dynamics of the load as
\begin{align}\label{eqn:load_dynamics}
	\begin{split}
	\ddconfigL &= \InertiaL^{-1}\left( -\coriolisL(\dconfigL) -\gravityL + \graspL(\configL)\cableForces \right),
	\end{split}
\end{align}
where  
$\configL = (\pL,\rotMatL)$;  
$\dconfigL = [\dpL^\top \vSpace {^L\angVelL}^\top]^\top$ with  $^L\angVelL \in \nR{3}$ the angular velocity of $\frameL$ w.r.t. $\frameW$ expressed in $\frameL$;
$\InertiaL = \diag{{\massL} \eye{3},\inertiaL}$ with $\eye{3} \in \nR{3 \times 3}$ the identity matrix; $\gravityL = [{\massL g} \vect{e}_3^\top \vSpace \vZero]^\top$, where $\rm{g}$ is the gravitational acceleration and $\vE{i}$ is the canonical unit vector with a~$1$ in the $i$-th entry. Coriolis and centrifugal terms are given by 
$$\coriolisL = \begin{bmatrix}\vZero \\ \vSpace \skew{\angVelL}\inertiaL\angVelL\end{bmatrix}$$
where $\skew{\star}$ is the \textit{skew operator}\footnote{Given $\vX \in \nR{3}$, $\skew{\vX} \in \nR{3 \times 3}$ is such that $\skew{\vX}\vY = \vX \times \vY$ for all  $\vY \in \nR{3}$.}, 
 and the grasp matrix is
\begin{align*}
	\graspL = \matrice{\eye{3} & \eye{3} \\ \skew{\anchorPosL{1}}\rotMatL^\top & \skew{\anchorPosL{2}}\rotMatL^\top}.
\end{align*}
%
The load is suspended by two cables from two anchoring points, $\anchorPoint{i}$ with $i = 1,2$, for which the position w.r.t. $\frameL$ is described by the vector $\anchorPosL{i} \in \nR{3}$. Each cable exerts on the load a force $\cableForce{i}$ such that
$\cableForces = [\cableForce{1}^\top \vSpace \cableForce{2}^\top]^\top$ in \eqref{eqn:load_dynamics}. 
By simple kinematics, the position of $\anchorPoint{i}$ w.r.t. $\frameW$ is then given by $\anchorPos{i} = \pL + \rotMatL\anchorPosL{i}$. 
Since we are considering a beam-like load, the object CoM is aligned with the two anchoring points of the cables.
Without loss of generality, we assume that $\anchorPosL{1}=[\anchorLength{1}\;0\;0]^\top$ and $\anchorPosL{2}=[-\anchorLength{2}\;0\;0]^\top$, where $\anchorLength{i} \in \nR{}_{>0}$, for $i = 1,2$. We also define the beam's  length $L=\anchorLength{1} + \anchorLength{2}.$  
\begin{figure}[t]
\center
\includegraphics[width=.8\columnwidth]{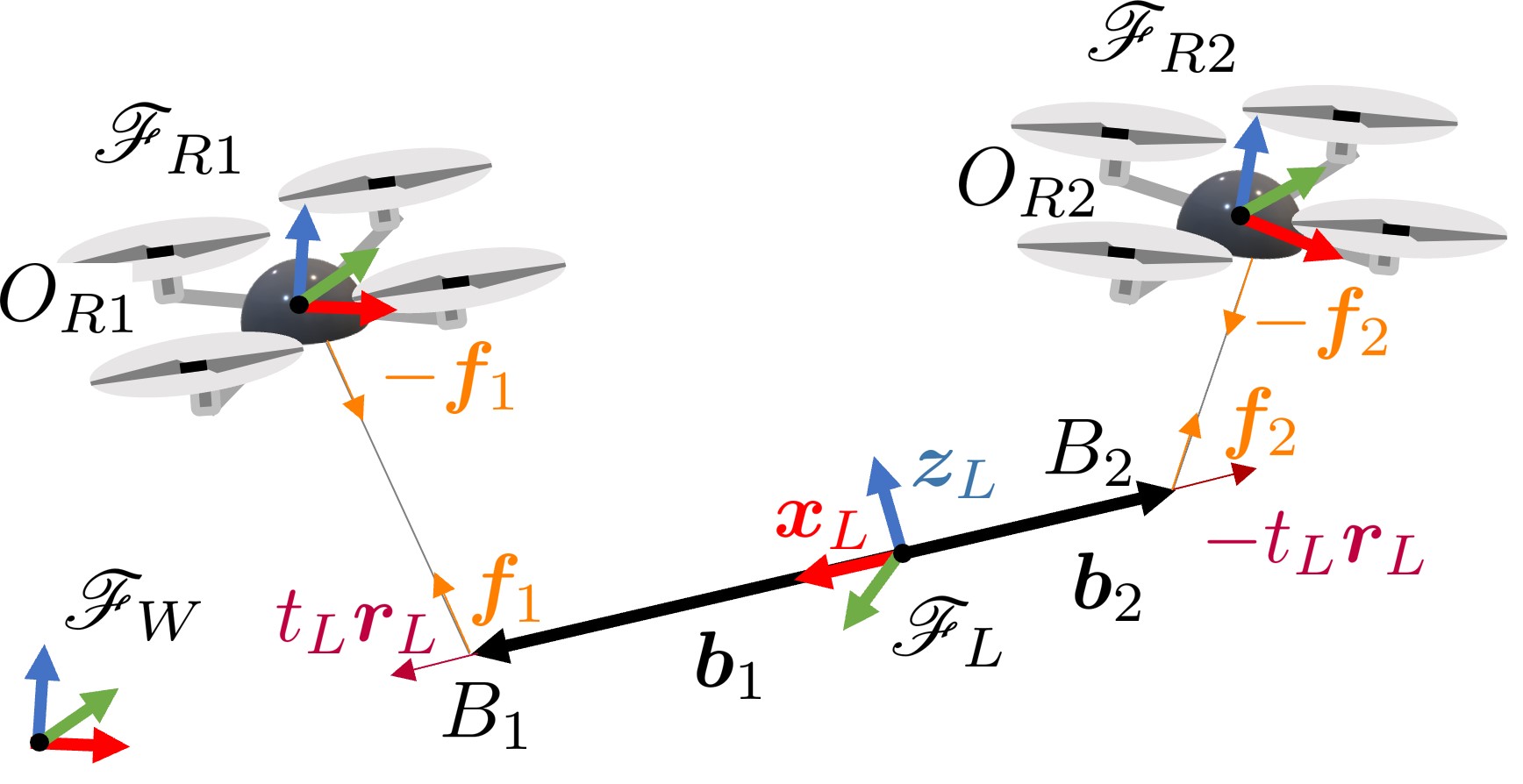}
\caption{Representation of the system and its major variables. The two aerial vehicles do not need to be necessarily quadrotors since the analysis and control design is valid for general aerial vehicles.}
\label{fig:model:systemDescription}
\end{figure}
\subsubsection{Robot Model}
 We define a frame ${\frameR{i} = \{ \originR{i},\xR{i},\yR{i},\zR{i}\}}$ rigidly attached to the i$-th$ robot and centered in its CoM. The $i$-th cable is attached to the $i$-th aerial vehicle at the point $\originR{i}$, which allows decoupling the robot's attitude dynamics from the rest \cite{2017-TagKamVerSieNie, goodman2022geometric}.
$\frameR{i}$ is used to describe the position and rotation of the vehicle  w.r.t. $\frameW$, denoted by the vector $\pR{i} \in \nR{3}$, and the rotation matrix $\rotMatR{i} \in \SO{3}$, respectively. 

The use of recent controllers for unidirectional- and multidirectional-thrust vehicles~\cite{kamel2017model,2016j-RylBicFra} and disturbance observers for aerial vehicles has been experimentally proven to result in negligible tracking errors even in the presence of external disturbances.

Consequently, due to the time-scale separation between the fast attitude dynamics and the slow translational dynamics~ \cite{nonami2010autonomous}, the closed-loop translational dynamics of the robot under the influence of the position controller effectively behaves like that of a double integrator
$ \ddpR{i} = \uR{i},$  
where $\uR{i}$ is a virtual input.
In other words, it is safe to assume that  the aerial robots together with a sufficiently accurate position controller can track any desired $C^2$ trajectory with negligible error in the domain of interest \cite{2018-RugLipOll}, independently from external disturbances~\cite{2018h-TogGabPalFra}. In this work, we follow such experimentally validated common practice for the theoretical derivations contained in Secs~\ref{sec:equilibria},~\ref{sec:stability}, and~\ref{sec:error} (see, e.g., the experiments on cooperative load transport in \cite{2017-TagKamVerSieNie, tagliabue2019robust}), and we utilize  underactuated quadrotors for both the numerical and experimental validations in  Sec~\ref{sec:exp_num} and Sec~\ref{sec:exp_real}. 


\subsubsection{Cable Model}
Cable-to-robot and cable-to-load connections are modeled as passive and mass-negligible rotational joints.
Besides, the $i$-th cable is represented as a unilateral spring along its principal direction, which is a frequently  adopted model~\cite{goodman2021geometric, 2018h-TogGabPalFra, bisig2020genetic, yiugit2021novel}. As commonly done in the state of the art, the  cables' mass and inertia are assumed negligible in comparison to the robots' and load's.
Its parameters are the constant elastic coefficient $\springCoeff{i} \in \nR{}_{> 0}$ and the constant rest length denoted by $\length{i}$.

The attitude of the $i$-th cable w.r.t. $\frameW$ is expressed by the normalized vector \footnote{$\mathbb{S}^{2} = \{ \vV \in \nR{3} \; | \; \norm{\vV} = 1 \}$} $\cableAttitude{i}/\norm{\cableAttitude{i}} \in \mathbb{S}^{2}$, where $\cableAttitude{i} = \pR{i} - \anchorPos{i}$.
The force acting on the load at $\anchorPoint{i}$, given a certain length $\norm{\cableAttitude{i}}$ of the cable, is given by the simplified Hooke's law:
\begin{align}\label{eqn:cableForce}
	\begin{split}
	\cableForce{i} = \norm{\cableForce{i}}\frac{\cableAttitude{i}}{\norm{\cableAttitude{i}}}\text{,} \;\;\;
	\norm{\cableForce{i}} = \begin{cases}
					{\springCoeff{i}}(\norm{\cableAttitude{i}} - \length{i})	& \text{ if } \norm{\cableAttitude{i}} - \length{i} > 0 \\
					0	&	\text{ otherwise }
				\end{cases}.
	\end{split}			
\end{align}
The force produced on the other hand of the cable, i.e., on the $i$-th robot at $\originR{i}$, is equal to $-\cableForce{i}$.

 
\begin{figure}[t]
	\center
	\includegraphics[width=0.8\columnwidth]{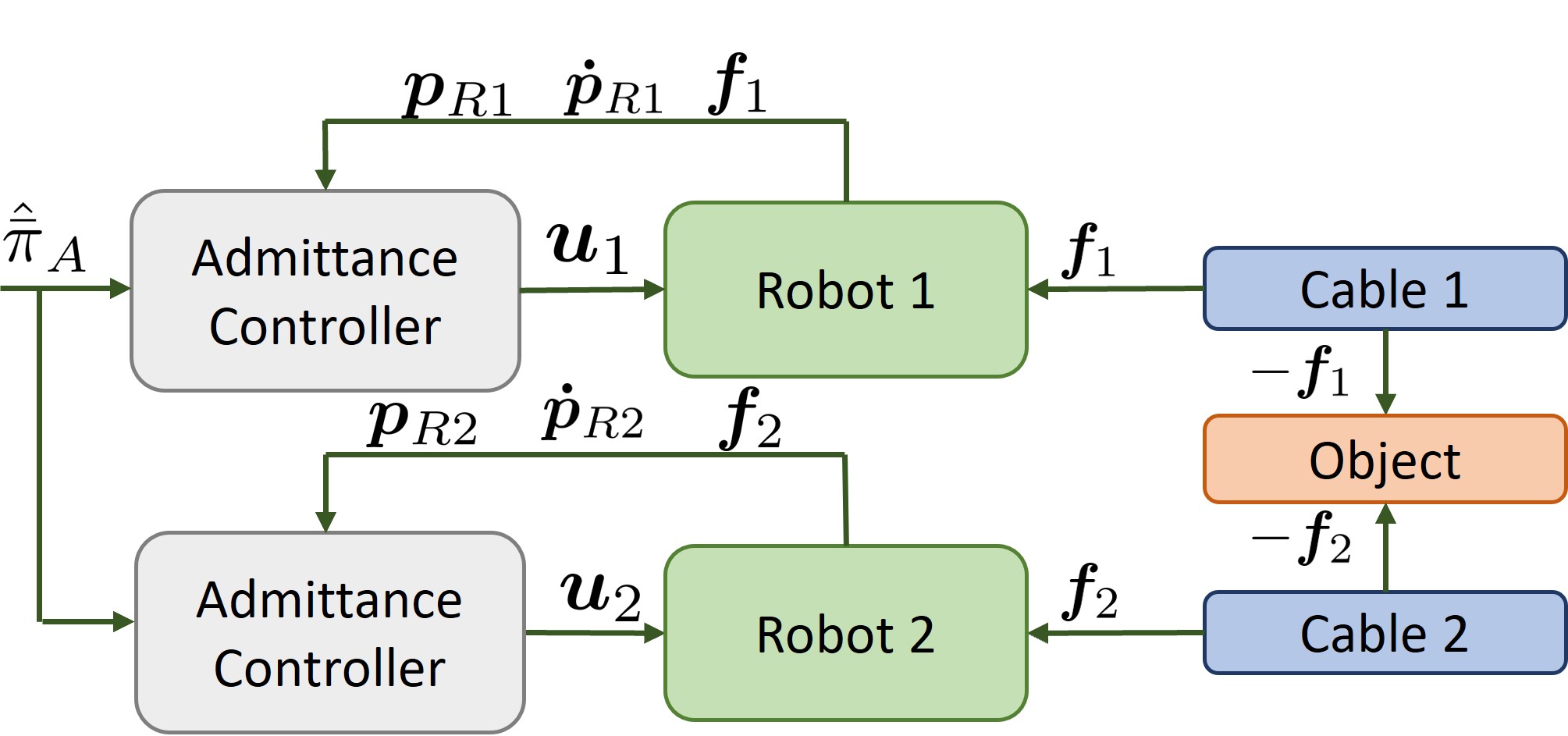}
	\caption{Schematic representation of the overall system including both physical and control blocks. The input of the admittance controller is affected by the uncertainty of  the system parameters.}
	\label{fig:control:controlStrategy}
\end{figure}
\subsubsection{Controller}
We recall that to regulate the pose of the manipulated load to a desired configuration  $\configLEq = (\pLEq, \rotMatLEq)$, an admittance controller is used on the robots~\cite{2018h-TogGabPalFra}:
\begin{align}\label{eqn:admittance}
	\uR{i} = \inertiaA{i}^{-1}\left( -\dampingA{i}\dpR{i} - \springA{i}\pR{i} - \cableForce{i} + \paramA{i} \right),
\end{align}
where the positive-definite symmetric matrices $\inertiaA{i},  \dampingA{i}, \springA{i} \in \nR{3 \times 3}$ are, respectively, the virtual inertia of the robot,  and the damping and stiffness coefficients of a virtual spring-damper system that links the robot and a desired reference frame; 
$\paramA{i} \in \nR{3}$ is an additional forcing input that is properly set to steer the load to the desired configuration.

\begin{rmk}
	One can notice that~\eqref{eqn:admittance} requires only local information, i.e. the robot's state $(\pR{i},\dpR{i})$, which can be retrieved with standard onboard sensors like IMU, GPS, and cameras; the force applied by the cable  $\cableForce{i}$,  which can be directly measured by an onboard force sensor or estimated by a sufficiently precise model-based observer as done in~\cite{2017e-RylMusPieCatAntCacFra,2017-TagKamVerSieNie}.
	Therefore, the described method is decentralized and does not require explicit communication between the robots. 
\end{rmk}
\subsubsection{Closed-loop Model}
From equations~\eqref{eqn:load_dynamics} and~\eqref{eqn:admittance}, the closed-loop system dynamics can be written as $\ddconfig = \dynamicModelFun(\config,\dconfig,\paramA{})$ where
\begin{align}\label{eqn:closedLoopDynamics}
	\dynamicModelFun(\config,\dconfig,\paramA{}) 
			   = \matrice{
			\inertiaA{}^{-1}\left( -\dampingA{}\dconfigR{} - \springA{}\configR{} - \cableForce{} + \paramA{} \right) \\
			\InertiaL^{-1}\left( -\coriolisL(\dconfigL) -\gravityL + \graspL\cableForces \right)},
\end{align}
with $\configR = [\pR{1}^\top \vSpace \pR{2}^\top]^\top,$  $\config = (\configR, \configL)$, $\dconfigR = [\dpR{1}^\top \vSpace \dpR{2}^\top]^\top$ and
$\dconfig = [\dconfigR^\top \vSpace \dconfigL^\top]^\top$;  
$\paramA{} = [\paramA{1}^\top \vSpace \paramA{2}^\top]^\top$.
Furthermore $\inertiaA{} = \text{diag}(\inertiaA{1},\inertiaA{2})$,
$\dampingA{} = \text{diag}(\dampingA{1},\dampingA{2})$ and 
$\springA{} = \text{diag}(\springA{1},\springA{2})$.

In order to coordinate the motion of the robots in a decentralized way,  a  \textit{leader-follower} approach is used. 
In this way, only the designated leader will have active control over the position of the load.
On the other hand,  the other robot will follow, partially sustaining the weight of the load and contributing to the control of the load attitude.
Choosing without loss of generality, robot~1 as the leader and robot~2 as the follower, the leader-follower approach is achieved as previously proposed in~\cite{2017-TagKamVerSieNie, tagliabue2019robust, 2018h-TogGabPalFra, gabellieri2020study} by setting  $\springA{1} \neq \vZero$ and $\springA{2} = \vZero$.

In the following, we present Theorem \ref{theom:equilibriumConfigToParam}, from \cite{2018h-TogGabPalFra}, along with two definitions that will help the reader comprehend the contribution of this work.
\begin{defin}[Equilibrium configuration]
	$\config$ is an \emph{equilibrium configuration}, indicated as  $\configEq$,  if $\exists$ $\paramA{}$ s.t. $\vZero = \dynamicModelFun(\config,\vZero,\paramA{}),$
i.e, if the corresponding zero-velocity state is a forced equilibrium for the system~\eqref{eqn:closedLoopDynamics} for a certain forcing input $\paramA{}$. 
\end{defin}
\begin{defin}[Load internal force]\label{def:internalForce}
	For the considered system,  the \emph{load internal force} is defined as
\begin{align}
\internalTension := \tfrac{1}{2}\cableForce{}^\top \matrice{\eye{3} \vSpace -\eye{3}}^\top\rotMatL\vE{1}
\label{eq:internal_force},
\end{align}
	where $\matrice{\eye{3} \vSpace -\eye{3}}^\top\rotMatL\vE{1}\in Ker(\graspL)$.
	We have that
	\begin{itemize}
		\item if $\internalTension > 0$ the internal force causes a \emph{tension} in the load;
		\item if $\internalTension < 0$ the internal force causes a \emph{compression}.
	\end{itemize}	 
\end{defin} 

The following result, proven in~\cite{2018h-TogGabPalFra}, provides the expression of the  forcing input $\paramA{}$ and the robot configurations $\configR$ for which, given a desired load configuration $\configLEq$, $\config = (\configR, \configLEq)$ is an equilibrium configuration of the system.
\begin{thm}[equilibrium inverse  problem, provided in~\cite{2018h-TogGabPalFra}, reported here for completeness]\label{theom:equilibriumConfigToParam}
Consider the  closed-loop system~\eqref{eqn:closedLoopDynamics} and assume that the load is at a given desired configuration $ \configLEq = (\pLEq, \rotMatLEq)$. 
For each internal force $\internalTension \in \nR{}$, there exists a unique constant value of the forcing input $\paramA{}= \paramAEq{}{}$ (and a unique position of the robots $\configR=\configREq$)
such that $\configEq = (\configLEq,\configREq)$ is an equilibrium of the system.

 In particular $\paramAEq{}$ and $\configREq = [\pREq{1}^\top\; \pREq{2}^\top]^\top$ are given by 
 \begin{align}
 		\paramAEq{}(\configLEq,\internalTension) &= \springA{}\configREq + \cableForcesEq(\configLEq,\internalTension) 
 		\label{eqn:equilibriumConditions:paramA}	\\
 		\pREq{i}(\configLEq,\internalTension) &= \pLEq + \rotMatLEq\anchorPosL{i} + \left(\frac{\norm{\cableForceEq{i}}}{\springCoeff{i}} + \length{i}\right) \frac{\cableForceEq{i}}{\norm{\cableForceEq{i}}},	\label{eqn:equilibriumConditions:pR}
 \end{align}
 for $i=1,2$, where
 \begin{align}
 	\cableForcesEq(\configLEq,\internalTension) = \matrice{\cableForceEq{1} \\ \cableForceEq{2}} = \matrice{\frac{\anchorLength{2}\massL g}{\rm{L}}\\\frac{\anchorLength{1}\massL g}{\rm{L}} }\matrice{\eye{3} \\ \eye{3}}\vE{3} + \internalTension\matrice{\eye{3} \\ -\eye{3}} \rotMatLEq\vE{1}. \label{eqn:equilibriumConditions:cableForces}
 \end{align}
\end{thm}

 From~\eqref{eqn:equilibriumConditions:paramA}, we can see that the forcing input is made up of two parts: one that depends on the robots' positions computed from the load equilibrium configuration according to kinematic relations, and the other one that depends on the equilibrium forces. The equilibrium forces are composed, according to \eqref{eqn:equilibriumConditions:cableForces}, by one term that compensates the gravity and one term that produces an internal force on the load whose intensity is $\internalTension$. \cite{2018h-TogGabPalFra} confirms that, if $\paramAEq{}$ is exactly applied to the closed-loop system \eqref{eqn:closedLoopDynamics},  $\configLEq$ is an isolated load equilibrium configuration if $\internalTension\neq0$, which is asymptotically stable if $\internalTension>0$ and unstable if $\internalTension<0$. Instead,  $\configLEq$ belongs to a continuum of equilibrium points containing any possible attitude of the load if $\internalTension=0$. 
In the remainder, Sec. \ref{sec:equilibria}-Sec. \ref{sec:error} contain the main theoretical contributions of the work.
\section{Equilibria under Uncertainty}\label{sec:equilibria}
In this section, the uncertainties are introduced and the equilibrium configurations of the system subject to those uncertainties are derived.
\begin{figure*}[t]
	\subfloat[][Representation of $\configSetEqZero$ when  $\condZero\neq0$. The position of the leader robot (in red) is the same in the two configurations. \label{fig:equilibrium:ZeroTi}]
	{\includegraphics[width=.28\textwidth]{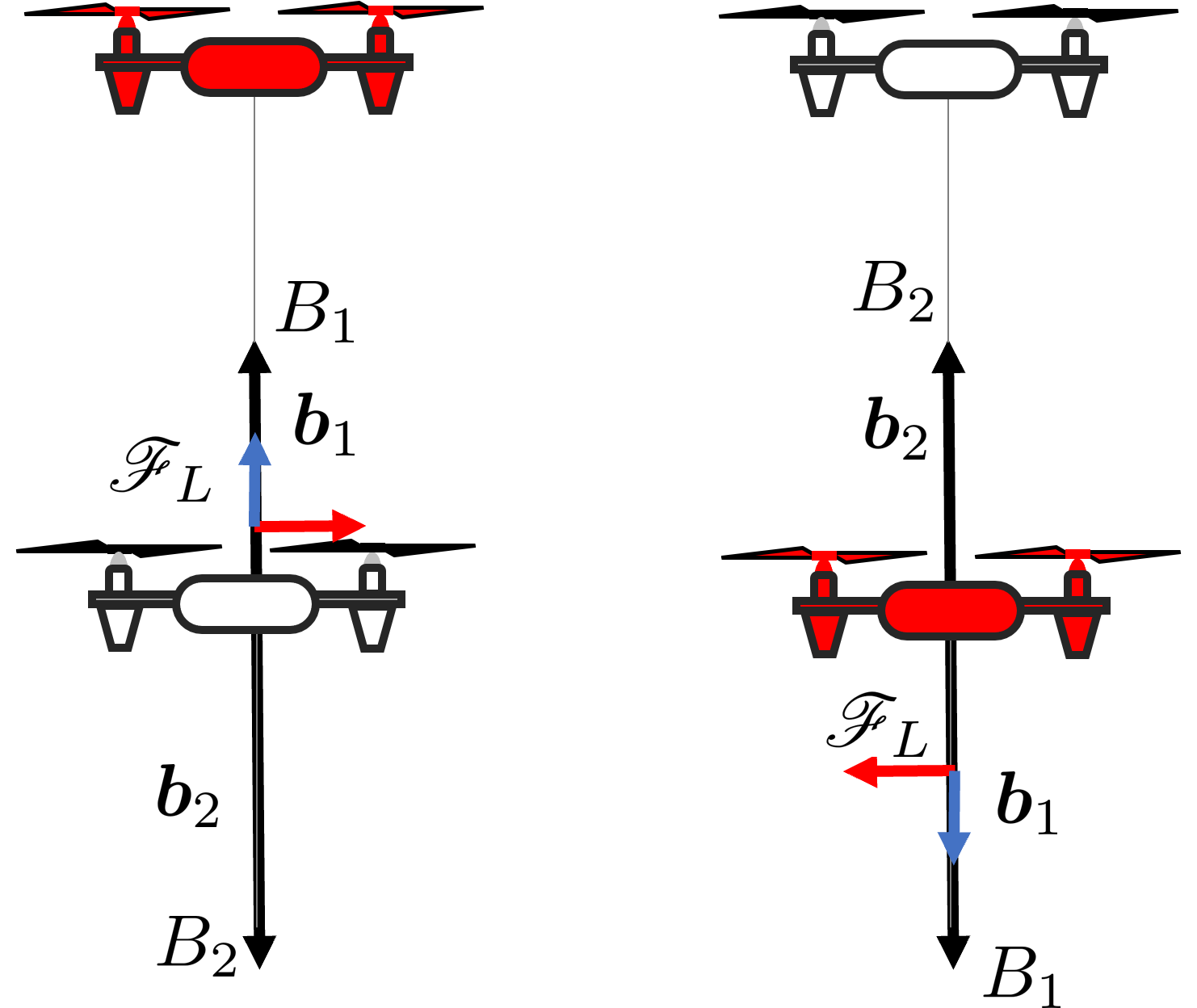}} \quad
		\subfloat[][Representation of $\configSetEqZero$ when  $\condZero=0$. The position of the leader robot (in red) is the same in all infinite equilibrium configurations, some of which are represented, and one of which is highlighted. Any other configuration with  $\anchorPoint{1}$ at the center of a sphere of radius $L$, $\anchorPoint{2}$ on its surface, and the cables vertical is in $\configSetEqZero$ when  $\condZero=0$.  \label{fig:equilibrium:ZeroTiZeroxi}]
	{\includegraphics[width=.40\textwidth]{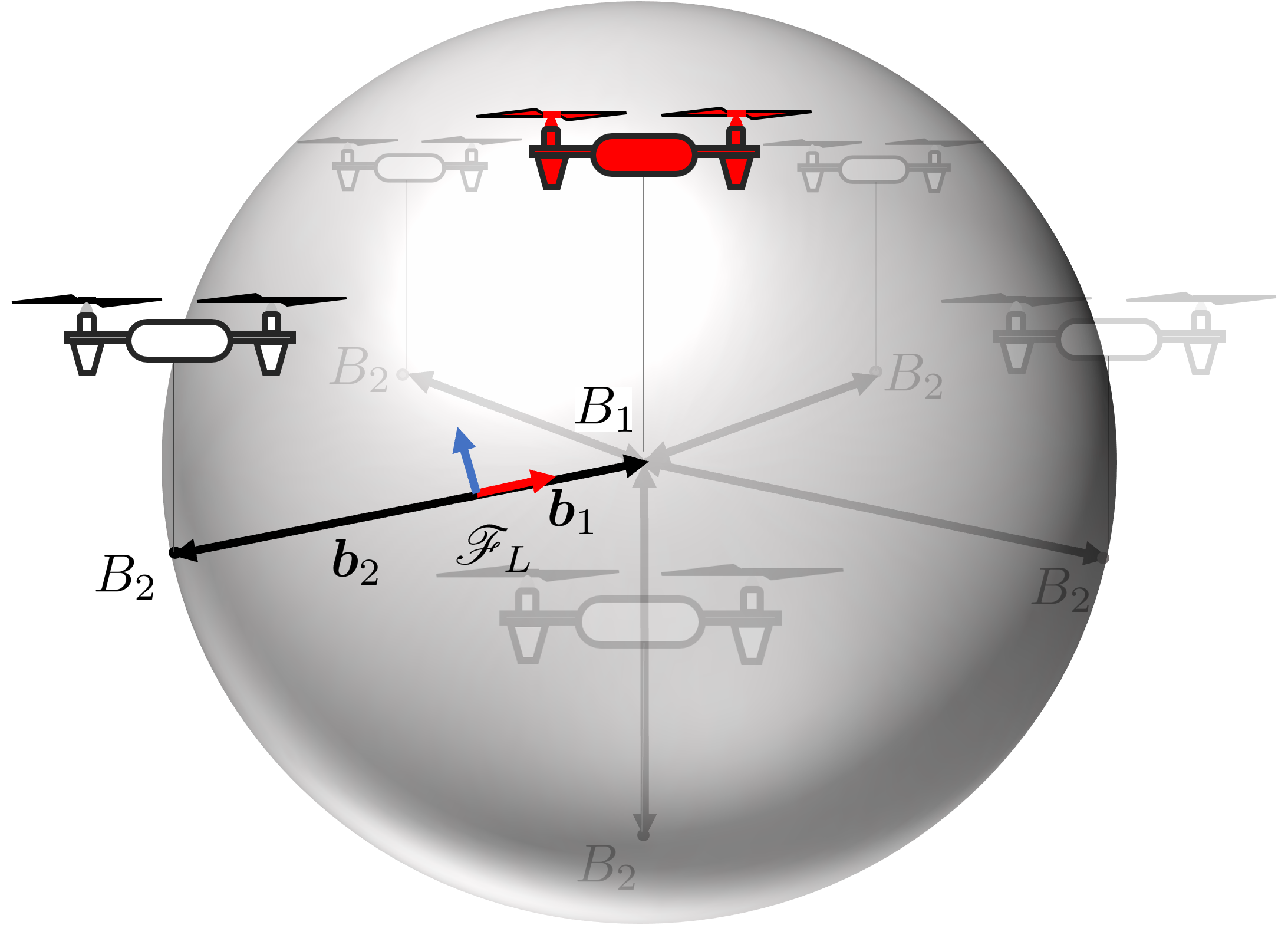}} \quad
	\subfloat[][Representation of $\configSetEq$ with  $\internalTension \neq 0$. The position of the leader robot (in red) is the same in the two configurations. On top, the load is under tension; below, it is under compression.  \label{fig:equilibrium:notZeroTi}]
	{\includegraphics[width=.28\textwidth]{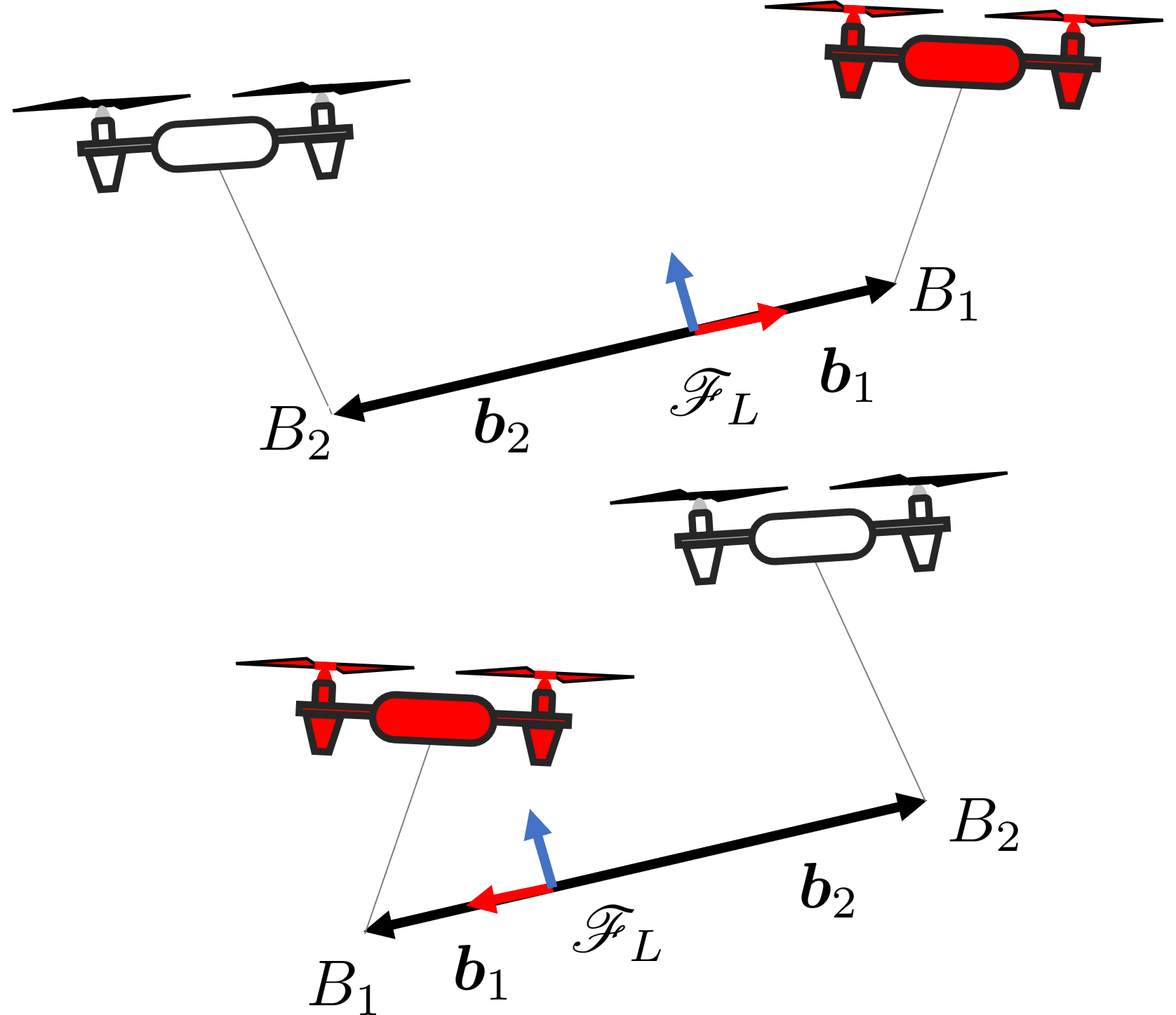}}\caption{Representation of the different equilibrium configurations of the system depending on  $\internalTension$ and $\condZero$.}
	\label{fig:subfig}
\end{figure*}
Note that, in reality,  $\paramAEq{}$ in \eqref{eqn:equilibriumConditions:paramA} cannot be applied exactly because of parametric uncertainties. Instead, one can apply only a version of $\paramAEq{}$, denoted with $\paramAEqU{}$,  computed using the nominal, uncertain values of the system parameters,  (see Figure~\ref{fig:control:controlStrategy} for a schematic representation of the control scheme with the nominal forcing input).  In the following, if not differently stated, we consider the general case in which a whole set of  uncertainties are present. These uncertainties affect the  control law \eqref{eqn:equilibriumConditions:paramA} and, in turn, affect the system equilibrium configurations. The uncertainties are the following:
\begin{itemize}
	\item $\massL$ is unknown, but only its nominal value $\massLU$ is available for the control design. We define the corresponding uncertainty as  $\Delta_m = \massL- \massLU$;
	\item $\anchorLength{1}$ is unknown, but only its nominal value ${\anchorLengthU{1}}$ is available. The corresponding uncertainty, affecting the load CoM position, is  $\Delta_b = \anchorLength{1} - {\anchorLengthU{1}}$;
	\item $\rm{L}$ is unknown, but only its nominal value $\hat{L}$ is available,  
	and we define $\Delta_\ell = \frac{1}{L}-\frac{1}{\hat{L}} = \ell- \hat{\ell}$ and $\Delta_L = L-\hat{L}$;
	\item the model of the cable $i$-th is inexact. Therefore,  the nominal length $\length{i}$ and stiffness $\springCoeff{i}$ are unknown, but their nominal values $ \lengthU{i}$ and ${\springCoeffU{i}}$ are available for the control design. We define the uncertainties $\Delta_{ki} = \springCoeff{i}- {\springCoeffU{i}}$, $\Delta_{\length{i}} = \length{i} - \lengthU{i}$.
	\end{itemize}
Note that the nominal value of $\anchorLength{2}$, $\anchorLengthU{2}$,  depends on the previously defined quantities according to the relationship  $\anchorLengthU{2}=\hat{L}-\anchorLengthU{1}$. However, for convenience, we also define $\Delta_{b2}=\anchorLength{2}-\anchorLengthU{2}$.

We shall now study the system's equilibrium configurations when ${\paramAEqU{}}$ is applied. 
%
%
\begin{thm}[equilibrium direct problem]
\label{theorem:paramToConfig}
	Given a desired load configuration $\configLEq = (\pLEq, \rotMatLEq)$  and the internal force $\internalTension \in \nR{},$ assume that  the forcing input $\paramAEqU{}{}$ 
	is computed
from~\eqref{eqn:equilibriumConditions:paramA} and is
	applied to the closed-loop system~\eqref{eqn:closedLoopDynamics}. Then, the equilibrium configurations are all and only the ones satisfying the following conditions:
	 \begin{align}
	 &\pR{1} ={\pREqInc{1}} -\springA{1}^{-1}(\Delta_m g\vE{3}):=\pREquilib{1}\label{pr1d3}\\
	 &\rotMatL:=\rotMatLEquilib~  {\rm{s.t.}}~  \skew{\vE{1}} {\rotMatLEquilib}^\top\Bigg[\left(\anchorLength{1}\massL - \frac{{\anchorLengthU{1}}{\massLU}L}{\hat{L}}\right)g\vE{3}   +\nonumber \\& + L\internalTension\rotMatLEq\vE{1} \Bigg] = \vect{0}\label{Rd3}\\
	 &\cableForce{1} = {\massL} g\vE{3} - \frac{\hat{\massL}{\anchorLengthU{1}}g}{\hat{L}}\vE{3}  + \internalTension{\rotMatLEq\vE{1}}:= {\cableForceEquilib{1}} \label{f1d3}\\
	 &\cableForce{2} =  \frac{\hat{\anchorLength{1}}{\massLU} g}{\hat{L}}\vE{3} - \internalTension\rotMatLEq\vE{1}= \cableForceEqInc{2}:=\cableForceEquilib{2}\label{f2d3}\\
	 &\pL = \pREquilib{1} -\rotMatLEquilib{\anchorPosL{1}} - \left(\frac{\norm{\cableForceEquilib{1}}}{{\springCoeff{1}}} + {\length{1}}\right) \frac{\cableForceEquilib{1}}{\norm{\cableForceEquilib{1}}}:=\pLEquilib\label{pld3},
	 \end{align}
	 where $\pREqInc{1}$ indicates the reference position of the leader robot computed as in \eqref{eqn:equilibriumConditions:pR}, namely starting from $\pLEq, \rotMatLEq$, but using the uncertain parameters.
	%
\end{thm}
	 \begin{proof}
	 $\paramAEqInc{}$ is defined according to \eqref{eqn:equilibriumConditions:paramA}, where \eqref{eqn:equilibriumConditions:cableForces} becomes
	\begin{align}
 	\cableForcesEqInc(\configLEq,\internalTension) = \matrice{\cableForceEqInc{1} \\ \cableForceEqInc{2}} = \matrice{\frac{(\hat{L}-{\anchorLengthU{1}}){\massLU} g}{\hat{L}}\\\frac{{\anchorLengthU{1}}{\massLU} g}{\hat{L}}}  \matrice{\eye{3} \\ \eye{3}}\vE{3} + \internalTension\matrice{\eye{3} \\ -\eye{3}} \rotMatLEq\vE{1}. \label{eqn:equilibriumConditions:cableForces_d3}
 \end{align}
 \begin{figure*}[t]
 	\centering
 	\subfloat[][$\internalTension>0$, $\condZero>0$. \label{fig:equilibriumPlos_tiPoscsipos}]
	{\includegraphics[width=.24\textwidth]{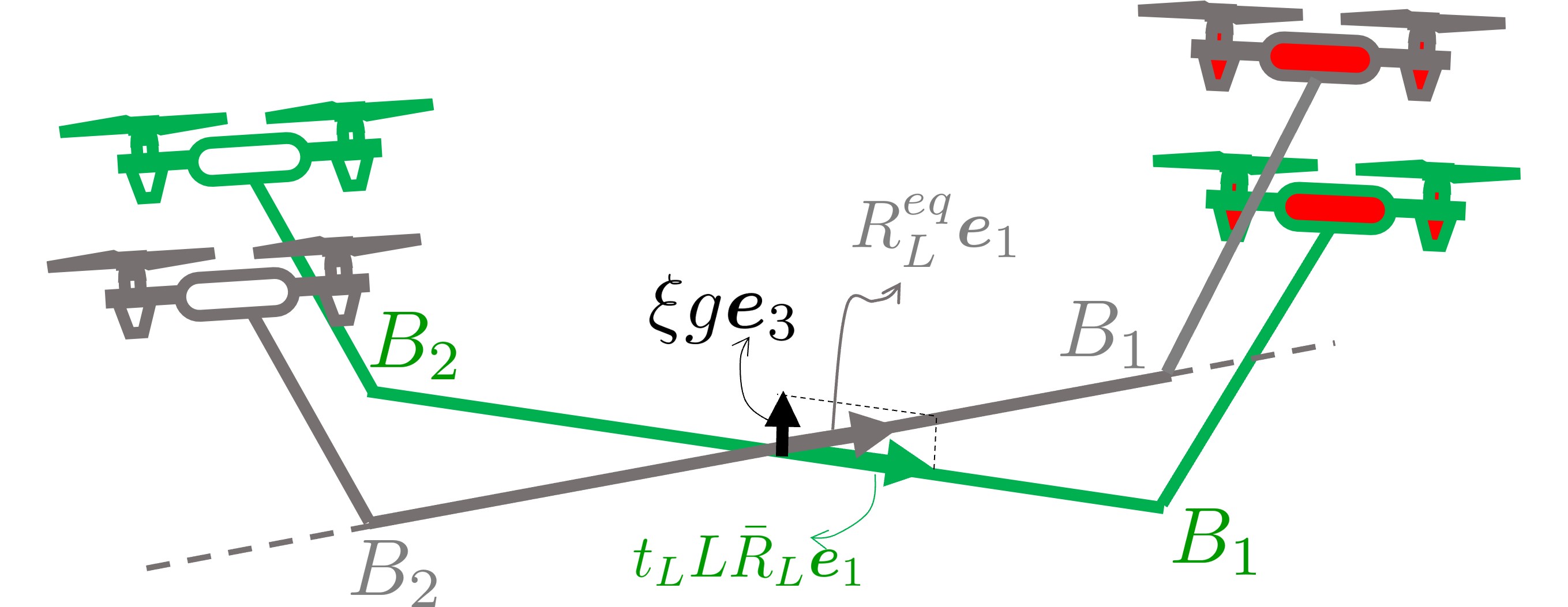}}
	\subfloat[][$\internalTension>0$, $\condZero<0$. \label{fig:equilibriumPlos_tiPoscsineg}]
	{\includegraphics[width=.24\textwidth]{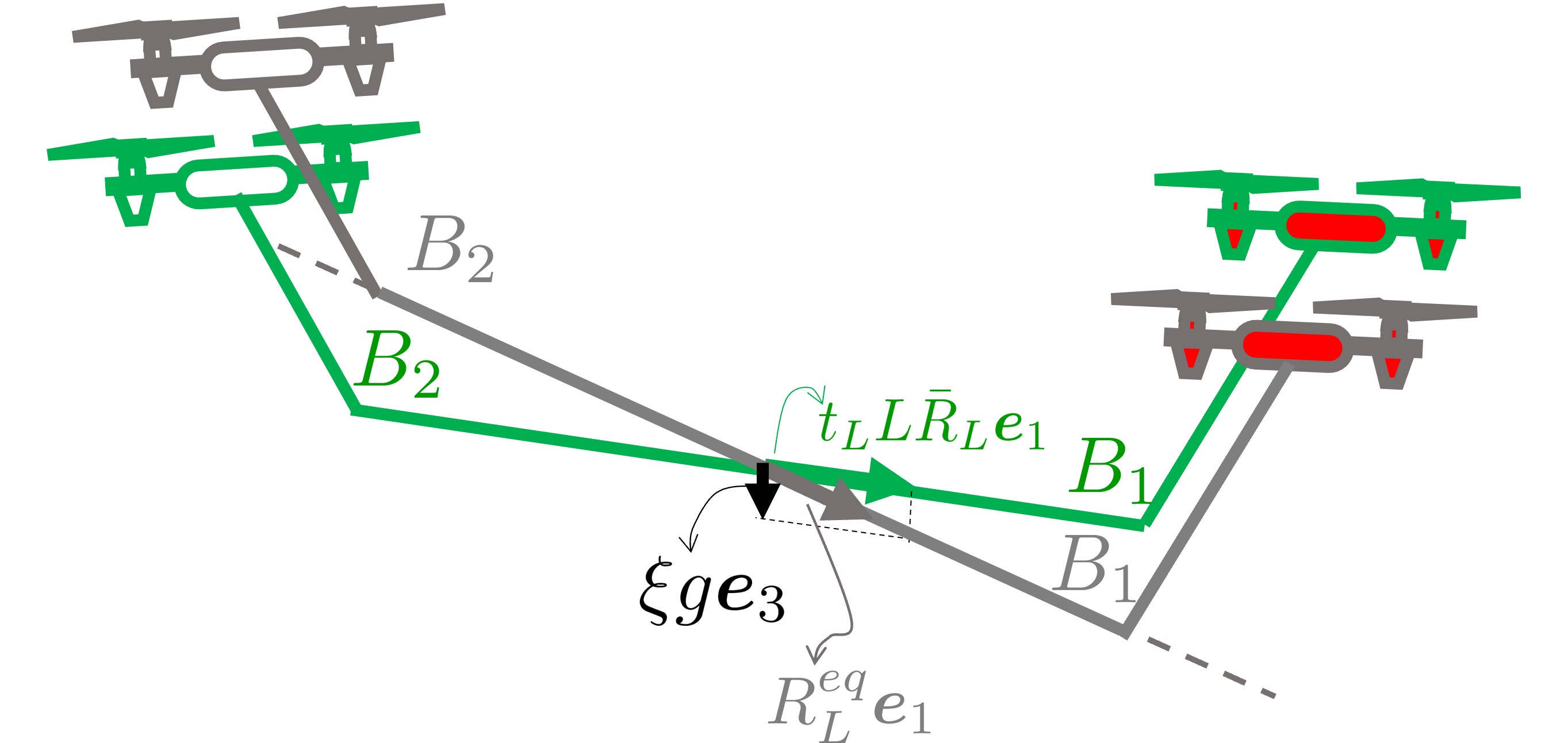}}
		\subfloat[][$\internalTension<0$, $\condZero>0$. \label{fig:equilibriumPlos_tinegcsipos}]
	{\includegraphics[width=.24\textwidth]{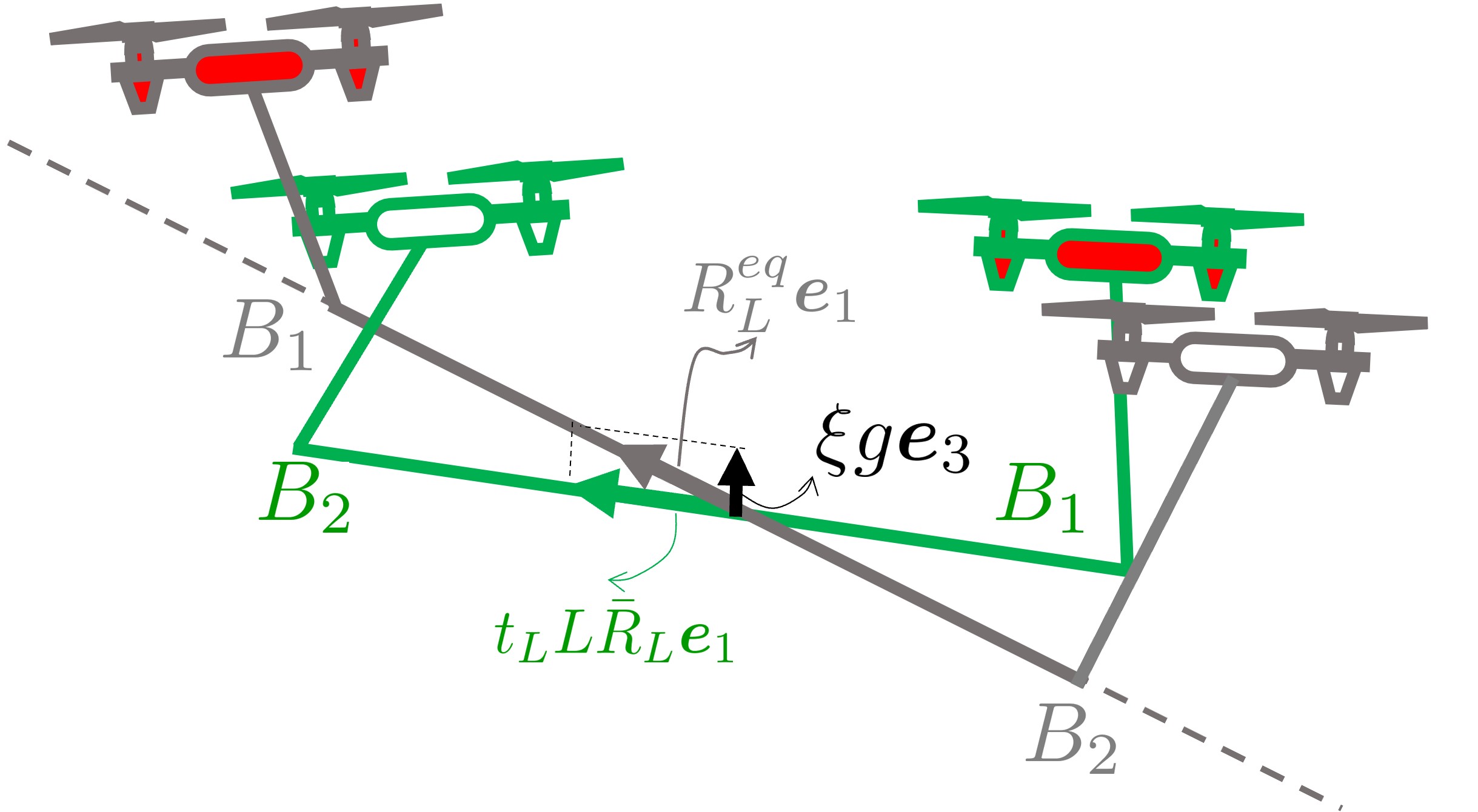}}
		\subfloat[][$\internalTension<0$, $\condZero<0$. \label{fig:equilibriumPlos_tinegcsineg}]
	{\includegraphics[width=.24\textwidth]{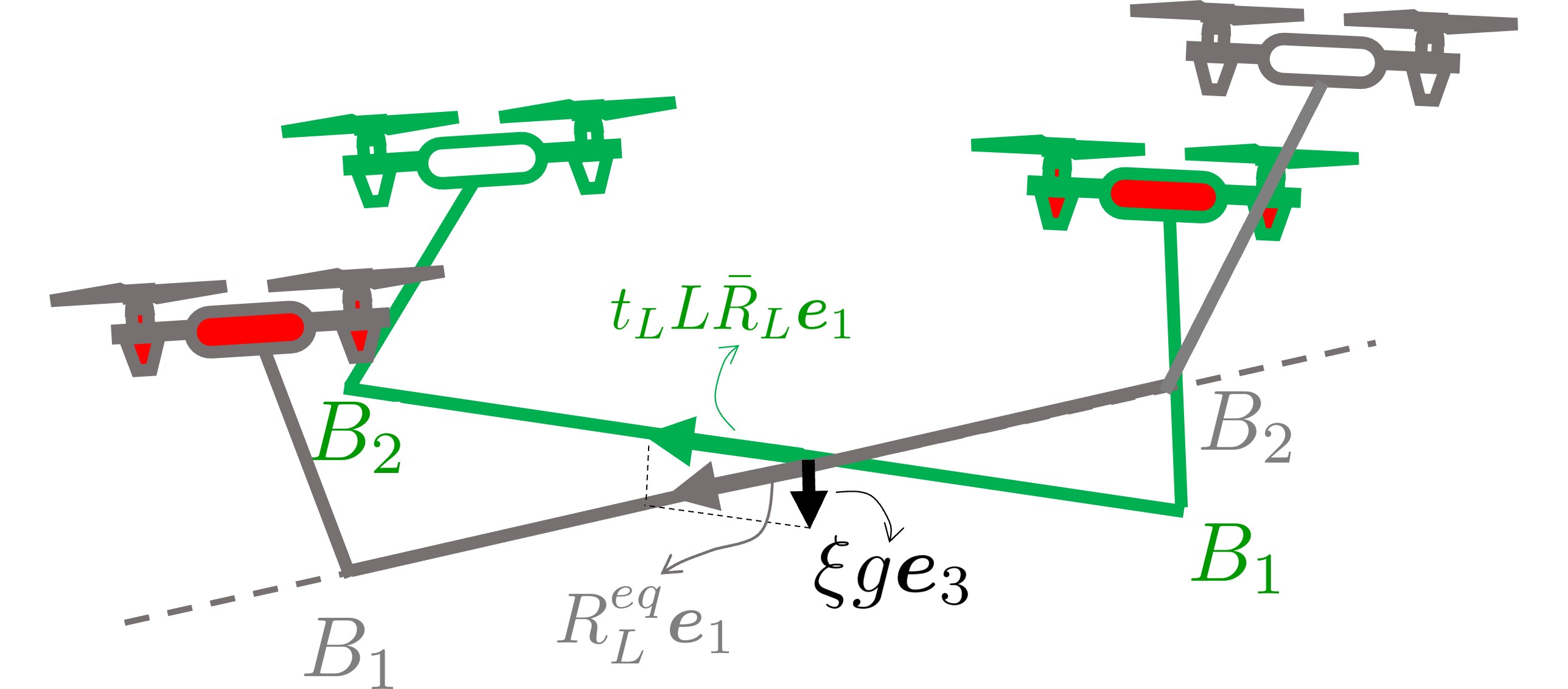}}
		\caption{Attitude of the load in the asymptotically stable equilibrium points $\stateSetEqPlus$ when $\internalTension\neq0$. In all the plots, the green configuration is the desired one  while in grey is the system with the actual  attitude at the equilibrium. In red are the leader robots. In $\stateSetEqPlus$, the attitude at the equilibrium is always such that $\rotMatLEquilib\vE{1}$ is aligned with $\internalTension\rotMatLEq\vE{1}$. However,  for $\internalTension<0$, this means the system is `flipped' compared to the desired configuration.}
	\label{fig:subfig2}
\end{figure*}
The control \eqref{eqn:admittance} is
 $$\uR{i} = \inertiaA{i}^{-1}\left( -\dampingA{i}\dpR{i} - \springA{i}\pR{i})-\cableForce{i} + \paramAEqInc{i} \right).$$
	Consider the equilibrium condition 
	\begin{equation}\label{eq:eq_conditions}
	\vect{0} = \dynamicModelFun(\config,0,\paramAEqInc{}).
	\end{equation}
	Equation \eqref{f2d3} is obtained by  substituting the last three lines of \eqref{eqn:equilibriumConditions:cableForces_d3} into \eqref{eq:eq_conditions} and solving the equilibrium condition for the follower robot. 
	Then, \eqref{f2d3} can be substituted into the load translational equilibrium (lines 7, 8, and 9 of \eqref{eq:eq_conditions}) to retrieve \eqref{f1d3}.
	\eqref{pr1d3} results from the first three lines of \eqref{eq:eq_conditions} using \eqref{f1d3}. Finally, \eqref{Rd3} can be obtained using \eqref{f1d3} and \eqref{f2d3} in the last three lines of \eqref{eq:eq_conditions}. Equation \eqref{pld3} is obtained applying the analogous of \eqref{eq:pL}.\end{proof}
\begin{defin}\label{def:configSetEq}
	Given a desired load configuration $\configLEq = (\pLEq, \rotMatLEq)$, internal force $\internalTension \in \nR{}$, and forcing input $\paramA{} = \paramAEqInc{}(\configLEq,\internalTension)$, we define the set of equilibrium configurations as $\configSetEq = \{ \config \text{ s.t. conditions of Theorem~\ref{theorem:paramToConfig} are satisfied} \}$
\end{defin}


From Theorem~\ref{theorem:paramToConfig}, we can distinguish between two scenarios:

\smallskip
	\emph{Scenario 1:}
	 If $\internalTension = 0$, condition~\eqref{Rd3} implies that the attitude of the load is such that $\rotMatLEquilib\vE{1}$ is aligned to $\vE{3}$, and conditions \eqref{f1d3} and \eqref{f2d3} imply that both cables are vertical. 
	In other words, the load at the equilibrium is, irrespective of the parametric uncertainties, aligned with the vertical direction; even an infinitesimal parametric uncertainty would lead the load to this undesired configuration in which the vertical load is aligned with the two vertical cables.  Such a configuration is clearly not realizable. Note that the position error of the system at the equilibrium still depends on the parametric uncertainties (see condition \eqref{pr1d3}). One can express the alignment between $\rotMatLEquilib\vE{1}$ and $\vE{3}$ as $(\rotMatLEquilib\vE{1})^\top\vE{3}=\pm1$. By convention, let us indicate with $\configSetEqZeroi{1}$ the system equilibrium configuration in which $(\rotMatLEquilib\vE{1})^\top\vE{3}=+1$ holds, namely the one in which the leader robot is above and the follower robot below, and with $\configSetEqZeroi{2}$ the other equilibrium configuration. 
Fig~\ref{fig:equilibrium:ZeroTi} illustrates the aforementioned equilibrium configurations. 
Note also that there is an additional possibility. 
With simple manipulation, remembering that $\anchorLengthU{1}=\anchorLength{1}+\Delta_b$, $\anchorLengthU{2}=\anchorLength{2}+\Delta_{b2}$, $\hat{L}=L+\Delta_L$, and defining $\condZero$ as follows,  the term in \eqref{Rd3} becomes: 
\begin{align}
	    &\condZero:=\left(\anchorLength{1}\massL - \frac{{\anchorLengthU{1}}{\massLU}L}{\hat{L}}\right) = \nonumber\\
	    &\Delta_m \anchorLength{1}+\frac{\Delta_{b2}}{\hat{L}}\massLU\anchorLength{1}-\frac{\Delta_{b1}}{\hat{L}}\massLU\anchorLength{2}. \label{eq:condzero}
	\end{align}
	If $\condZero=0$, \eqref{Rd3} is verified for every value of $\rotMatL$, and hence the equilibrium configurations $\configSetEqZeroi{}$ are infinite and such that the attitude of the load at equilibrium is arbitrary. This happens in the special case in which the parameters of the system are exactly known (this situation is the one we analyzed in~\cite{2018h-TogGabPalFra} and which we can now see as a  special case with $\condZero=0$). Indeed, $\condZero=0$ is verified also if the cable parameters are the sole uncertain ones, as it will be also more deeply discussed in the following. See \ref{fig:equilibrium:ZeroTiZeroxi} for a schematic representation of the mentioned equilibrium configurations. 

\smallskip
\emph{Scenario 2:}
  If $\internalTension \neq 0$, condition~\eqref{Rd3} holds when the vectors $
\rotMatLEquilib\vE{1}$ and
%
\begin{equation}\left(\condZero g\vE{3} + L\internalTension\rotMatLEq\vE{1} \right)\label{eq:vec2}\end{equation}are aligned. Similar to before, this condition holds in two possible cases: when the vectors are aligned and point in the same direction, or when they are aligned but point in opposite directions. 
Let us indicate with ${\rotMatLEquilib}^+$ the attitude of the load for which condition \eqref{Rd3} holds and the two vectors $
\rotMatLEquilib\vE{1}$ and \eqref{eq:vec2} point in the same direction.  We indicate the corresponding load equilibrium configuration as $\configSetEqPlus$. In the other case, when the two aforementioned vectors point in opposite directions, at the equilibrium one has that  ${{\rotMatLEquilib}^-={\rotMatLEquilib}^+\rotMatVectAngle{\zL}{\pi}}$; we indicate the corresponding equilibrium configuration as $\configSetEqMinus$. Depending on the sign of $\internalTension$ in $\paramAEqInc{}$, the forces in the cables place the load under tension in one equilibrium configuration and under 
compression in the other. Figure \ref{fig:equilibrium:notZeroTi} represents these equilibrium configurations. 
	\begin{rmk}\label{rmk:dm_tl}	Under the hypothesis that $\bar{\pitch}\neq \pi/2 + k\pi$, with ${k\in\mathbb{N}}$  and ${\internalTension\neq 0}$,as shown by~\eqref{Rd3}, at the equilibrium the following holds:
	\begin{align}
	\yaw &= \bar{\yaw} + k\pi \label{eq:inc_mass_yaw}\\
	\tan{\pitch} &= \tan{\bar{\pitch}} + \frac{-\condZero g}{L\internalTension\cos{\bar{\pitch}}}. \label{eq:inc_mass_pitch} 
	\end{align}
	In other words, \emph{the uncertainties have no effect on the yaw angle at equilibrium.} $\yaw$ may differ from $\bar{\yaw}$ by $\pi$ because, as already discussed, both $\configSetEqPlus$ and $\configSetEqMinus$ are equilibrium configurations.
	Moreover, \eqref{eq:inc_mass_pitch} tells us that not only is the attitude error proportional to the amount of uncertainty but also that, as $\internalTension$ decreases, the load at the equilibrium becomes \textit{increasingly vertical}.  
Eventually, for $\internalTension =0$ and uncertain parameters ($\condZero\neq 0$),~\eqref{Rd3} leads to
$ \vE{1} \times \rotMatL^\top \vE{3} = \vect{0}.
$
Namely, as previously observed, the load at the equilibrium is aligned with the vertical direction and the two cables are vertical despite the value of $\condZero\neq0$.  In other words, if $\internalTension=0$ the load attitude error is unaffected by the parametric uncertainties: the load will reach the same, clearly undesired, configuration regardless of the smallest  $\condZero\neq0$.
\end{rmk}

	In the remainder of this section, we briefly analyze the effects of each uncertain parameter on the final equilibrium.
	\subsection{Uncertainty on the load mass $\massL$}
	In this subsection, we only discuss uncertainty in the load's mass, while the other parameters are assumed to be  perfectly known. 
	Equations \eqref{pr1d3}-\eqref{f2d3} become:
	\begin{align}
	&\pREquilib{1} = {\pREqInc{1}}-\springA{1}^{-1}\Delta_m g\vE{3}\label{pr1dm}\\
&{\anchorLength{1}}\skew{\vE{1}}{\rotMatLEquilib}^\top g{\Delta_m}{\vE{3}} + \internalTension L \skew{\vE{1}}{\rotMatLEquilib}^\top\rotMatLEq\vE{1}=\vect{0} \label{Rdm}\\
&\cableForceEquilib{1} = \massL g\vE{3} -\frac{{\anchorLength{1}}{\massLU}g}{L} \vE{3}+ \internalTension \rotMatLEq\vE{1} = \cableForceEqInc{1} + \Delta_m g \vE{3} \label{f1dm}\\
&\cableForceEquilib{2} = \frac{\anchorLength{1} {\massLU g}}{L} \vE{3} - \internalTension \rotMatLEq\vE{1} = \cableForceEqInc{2} \label{f2dm}.
	\end{align}

The position of the load CoM at the equilibrium is different from $\pLEq$ and can be computed from  \eqref{pld3} using  \eqref{pr1dm}-\eqref{f1dm}. 

It is worth noting that the leader robot can detect a mismatch between the known commanded 
$\cableForceEqInc{1}$ and the actual force $\cableForceEquilib{1}$ measured at steady state. Such a discrepancy solely depends on $\Delta_m$, according to \eqref{f1dm}.
 Thus, the leader robot can compute $\Delta_m$ and, by knowing the nominal value ${\massLU}$, retrieve the actual value of the load mass $\massL$, which can be used to adjust its own reference force and position. 
 However, note that in a communication-less setup it is impossible for both robots to know the correct parameter value based simply on their own state. In fact, according to \eqref{f2dm}, the follower robot has no mismatch between the equilibrium and the force reference value.

\subsection{Uncertainty on the load length, $\rm{L}$, or CoM position, $\anchorLength{1}$}
Uncertainties on one of these two parameters have similar effects.
In one case, $\anchorLengthU{1}\neq \anchorLength{1}$, namely the load CoM is aligned to the cables attachment points on the load at an uncertain position but $L$ is exactly known; in the other case, $\anchorLength{1}$ is  exactly known but $L$ is not. 
	In both cases, at the equilibrium, the following conditions hold:
	\begin{align}
	&\pREquilib{}{1} = \pREqInc{1}\label{pr1db}\\
&\skew{\vE{1}} {\rotMatLEquilib}^{\top} ({\rm{L}\internalTension} \rotMatLEq\vE{1} + y {\massL} g \vE{3}) =\vect{0} \label{Rdb}\\
&\cableForceEquilib{1}  = \cableForceEqInc{1}\label{f1db}\\
&\cableForceEquilib{2} = \cableForceEqInc{2}\label{f2db},
	\end{align}
where $y = \Delta_b$ in one case, and $y=\anchorLength{1} \rm{L} \Delta_\ell$ in the other. $\pREqInc{1},\cableForceEqInc{1}$, and $\cableForceEqInc{2}$ are computed from \eqref{eqn:equilibriumConditions:pR} and \eqref{eqn:equilibriumConditions:cableForces}, where the corresponding uncertain parameter is used in place of the real one. 
	 	
	
In this case, the leader robot position and both cable forces at the equilibrium coincide with the respective reference values available to the robots (see~\eqref{pr1db}-\eqref{f2db}). Consequently, \emph{it is not possible for any of the robots to estimate the uncertain parameter at the equilibrium based on the local information they possess.}

	\subsection{Uncertainty on the cable length $\length{i}$ or stiffness $\springCoeff{i}$}
	Consider an uncertainty on the parameters of the i$-{th}$ cable such that the rest length is $ \length{i}\neq{\lengthU{i}}$ and the stiffness is $\springCoeff{i}\neq {\springCoeffU{i}}$. At the equilibrium,  $\cableForceEquilib{i} = \cableForceEq{i}$,  \begin{align}
	    \rotMatLEquilib = \rotMatLEq. \label{eq:Rlinc} 
	\end{align}  \begin{align}\label{pr1dlo}
	\pREquilib{1} = \pREqInc{1} = \pLEq + \rotMatLEq \anchorPosL{1}  + \left(\frac{\norm{\cableForceEq{1}}}{{{\springCoeffU{1}}}} + {{\lengthU{1}}}\right) \frac{\cableForceEq{1}}{\norm{\cableForceEq{1}}},
	\end{align} 
	and the value of $\pL$ at the equilibrium is 
	\begin{align}\label{pLeqdl0}
	&\pLEquilib = \pREqInc{1} -\rotMatLEq\anchorPosL{1} - \left(\frac{\norm{\cableForceEq{1}}}{{\springCoeff{1}}} + {\length{1}}\right) \frac{\cableForceEq{1}}{\norm{\cableForceEq{1}}}
	\end{align}

	We highlight that knowledge about the cable properties is required only when computing the reference position of the leader robot, according to \eqref{eqn:equilibriumConditions:pR}. What is more, only the information about $\springCoeff{1}$ and $\length{1}$ is  required. We conclude that  knowledge of \emph{ $\length{2}$ and $\springCoeff{2}$, is not necessary to stabilize the load at a desired pose. Moreover, $\length{1}$ and $\springCoeff{1}$ have no effect on the load attitude at equilibrium but they do influence the load position.} This is evident by simply substituting~\eqref{pr1dlo} into \eqref{pLeqdl0} with $\length{1}\neq\lengthU{1}$ and $\springCoeff{1}\neq\springCoeffU{1}$. Note that, since the robots' forces and the leader robot's position at the equilibrium coincide with the reference values available to the robots themselves, they are unaware of the load pose error induced by this uncertainty. 
\section{Stability Analysis}\label{sec:stability}
In this section, we shall analyze the stability of the equilibrium configurations discovered in Sec~\ref{sec:equilibria}.
First, being ${\state = (\config,\dconfig)}$ the state of the system, we define the following equilibrium states (subspaces of the state space): 
\begin{itemize}
		\item $\stateSetEqZeroi{} = \{ \state \; : \; \config \in \configSetEqZeroi{}, \; \dconfig = \vZero \}$,
	\item $\stateSetEqZeroi{1} = \{ \state \; : \; \config \in \configSetEqZeroi{1}, \; \dconfig = \vZero \}$,
	\item $\stateSetEqZeroi{2} = \{ \state \; : \; \config \in \configSetEqZeroi{2}, \; \dconfig = \vZero \}$,
	\item $\stateSetEqPlus = \{ \state \; : \; \config \in \configSetEqPlus, \; \dconfig = \vZero \}$,
	\item $\stateSetEqMinus = \{ \state \; : \; \config \in \configSetEqMinus, \; \dconfig = \vZero \}$.
\end{itemize}
\begin{thm}\label{thm:stability}
	Let us consider a desired load configuration $\configLEq$. For the system~\eqref{eqn:closedLoopDynamics}, let the constant forcing input be $\paramAEqInc{}$. 
	Then, 
	\begin{itemize}
		\item $\stateSetEqZeroi{1}$ is asymptotically stable if $\condZero>0$ 
		and unstable if $\condZero<0$;
		\item $\stateSetEqZeroi{2}$ is asymptotically stable if $\condZero<0$ 
		and unstable if $\condZero>0$.
		\item $\stateSetEqZeroi{}$ is a set of marginally stable equilibrium points if $\condZero=0$. 
		\item $\stateSetEqPlus$ is asymptotically stable
		\item $\stateSetEqMinus$ is unstable.
	\end{itemize}
\end{thm}
\begin{proof}
Consider the following Lyapunov candidate function:
	\begin{align}
		\label{eqn:LyapunovFunction}
		V(\state ) =& \frac{1}{2} (\dconfigR^\top \inertiaA{} \dconfigR +  \errorpREq{}^\top\springA{}\errorpREq{} + \dconfigL^\top\InertiaL\dconfigL + \nonumber \\
		&  {+\springCoeff{1}}  {(\norm{\cableAttitude{1}}-\length{1})^2}  +  {\springCoeff{2}} {(\norm{\cableAttitude{2}} - {\length{2}})^2)}  -{ \cableAttitude{1}^{\top} \cableForceEquilib{1}  } + \nonumber\\& - \cableAttitude{2}^\top\cableForceEquilib{2} +V_0 +\Vadd, 
	\end{align}
	where the robot position error is $\errorpREq{}=\pR{}-\pREquilib{}$, $V_0$ is constant, and $\Vadd$ is an additional term explained in the following. Function \eqref{eqn:LyapunovFunction} is composed of  standard positive definite quadratic terms equal to zero in the equilibrium points and by two terms of the form ${\frac{1}{2}{\springCoeff{i}} (\norm{\cableAttitude{i}}-{\length{i}})^2)  -{\cableAttitude{i}}^\top\cableForceEquilib{i}}$, call them $V_i(\state)$: these are  linked to the elastic energy of the cables and  have a minimum at the equilibrium as well.  A detailed proof of the former point can be found in~\cite{2018h-TogGabPalFra}. The proof first shows that  $V_i(\state)$ is \textit{radially unbounded}, i.e., $\lim_{\norm{\state}\to\infty} V_i(\state) = \infty$. Then, based on this result and Theorem 1.15 of~\cite{2000-HorParVan}, the term has a global minimum. Finally, it has  been shown that the global minimum of $V_i(\state)$ corresponds to the considered equilibrium~\cite{2018h-TogGabPalFra}. 
	
	We define the value of $V_i(\state)$ at the equilibrium (its minimum value) as $-V_0$, and we cancel it in \eqref{eqn:LyapunovFunction} so that its value at the equilibrium is zero. 

 Let us start considering  $\stateSetEqZeroi{1}$ and $\condZero>0$. In this case, we set $\Vadd = \condZero g (1-\vE{3}^\top\rotMatL\vE{1})$. With this choice,  \eqref{eqn:LyapunovFunction} is zero in $\stateSetEqZeroi{1}$ because also the term $1-\vE{3}^\top\rotMatL\vE{1}$ is zero in $\stateSetEqZeroi{1}$ by definition (load aligned with the vertical with $\rotMatL\vE{1}$ and $\vE{3}$ pointing in the same direction) and positive elsewhere (the scalar product $\vE{3}^\top\rotMatL\vE{1}\leq1$ because $\vE{3}$ and $\rotMatL\vE{1}$ have both unit norm).


Studying the sign of the time derivative of~\eqref{eqn:LyapunovFunction},  using~\eqref{eqn:closedLoopDynamics},~\eqref{eqn:cableForce}, and~\eqref{eqn:equilibriumConditions:cableForces},  we obtain ${\dlyapunovFun  = -\dconfigR{}^\top\dampingA{}\dconfigR{}},$ which is clearly negative semidefinite. In particular, let us define $\dVZeroSet= \{\state \; : \; \dlyapunovFun =0\}$. In this case, we have $\dVZeroSet = \{\state \; : \; \dconfigR = \vZero,\; \angVelL = \vZero \}$.
	
	Since $\dlyapunovFun$ is only negative \emph{semi}definite, we rely on  \textit{LaSalle's invariance principle} to complete the proof: one can easily verify from \eqref{eqn:closedLoopDynamics} that the largest invariant set in $\dVZeroSet$  is $\stateSetEqZeroi{1}$. 
	
	Analogous reasoning can be used when $\condZero<0$. The computation of $\dot{V}$ does not change, and it is, thus, negative semidefinite. However,  $\stateSetEqZeroi{1}$ is a set of accumulation for the points where $V(\state) < 0$ if $\condZero<0$.
	To see this, consider $\dconfig = \vZero$ and all quantities at the equilibrium apart from $\rotMatL$, which is such that $\vE{3}^\top\rotMatL\vE{1} = 1 - \epsilon$, with $\epsilon > 0$ arbitrarily small, meaning that $\rotMatL$ is arbitrarily close to $\rotMatLEquilib$.
	Under these conditions, $V(\state) =  g\condZero\epsilon < 0$.
	All conditions of \textit{Chetaev's theorem}~(the formulation of both this and La Salle's invariance principle can be found, e.g., in~\cite{khalil2002nonlinear}) are satisfied. Hence, we can conclude that $\stateSetEqPlus$ is unstable.
To show that $\stateSetEqZeroi{2}$ is asymptotically stable if $\condZero<0$, we set $\Vadd=-\condZero g (1+\vE{3}^\top\rotMatL\vE{1})$, which is zero in $\stateSetEqZeroi{2}$ (when  $\vE{3}^\top\rotMatL\vE{1}=-1$ by definition), and positive elsewhere.  The same Lyapunov candidate function is used to show that $\stateSetEqZeroi{2}$ is unstable if $\condZero>0$.  The reasoning is exactly dual to the previous case, hence it is here omitted for the sake of space.

Consider now $\condZero=0$. We set $\Vadd=0$, so that \eqref{eqn:LyapunovFunction}  is zero in  $\stateSetEqZeroi{}$ and positive elsewhere. Moreover, $\dot{V}$ is negative semidefinite as before. One can easily show that the largest invariant set is  $\stateSetEqZeroi{}$. We can thus say that the system state converges to a state $\state\in\stateSetEqZeroi{}$, which is, however, composed of a continuum of equilibrium points. Hence they are only \textit{marginally} stable. 

Finally, we study the stability of the equilibrium points when $\internalTension\neq 0$. In this case, 
 we set ${\Vadd=-(\condZero g\vE{3} +\internalTension L \rotMatLEq\vE{1})^\top {\rotMatL}\vE{1}+V_0',}$ with $V_0'= (\condZero g\vE{3} +\internalTension L \rotMatLEq\vE{1})^\top {\rotMatLEquilib}\vE{1}$. Clearly, $\Vadd$ and hence $\lyapunovFun$ are zero at the equilibrium. Moreover, $\Vadd$ is positive elsewhere by definition of $\stateSetEqPlus$, ($\condZero g\vE{3} +\internalTension L \rotMatLEq\vE{1}$ and $\rotMatL\vE{1}$ are aligned and point in the same direction when $\rotMatL=\rotMatLEquilib$, so that $\Vadd$ has its minimum in $\stateSetEqPlus$). Moreover, $\dlyapunovFun=-\dconfigR{}^\top\dampingA{}\dconfigR{}$, and the application of LaSalle's invariance principle leads to the conclusion that $\stateSetEqPlus$ is an asymptotically stable equilibrium point, similarly to before. To show the instability of $\stateSetEqMinus$, we use the same choice for $\Vadd$. However, since $\condZero g\vE{3} +\internalTension L \rotMatLEq\vE{1}$ and $\rotMatL\vE{1}$ are anti-parallel in  $\stateSetEqMinus$, $\Vadd$ is still zero at the equilibrium but negative when $\rotMatL
$ is arbitrarily close to $\rotMatLEquilib$. $\stateSetEqMinus$ is a point of accumulation for the points in which $\dlyapunovFun$ is negative, while $\dlyapunovFun$ remains negative semi-definite. For Chetaev's theorem, we conclude that $\stateSetEqMinus$ is unstable. \end{proof}
It is important to highlight that, as shown in Figure \ref{fig:subfig2}, for $\internalTension>0,$ $\stateSetEqPlus$ corresponds to a configuration of the system in which $\rotMatLEquilib$ (irrespective of the sign of $\condZero$) is the closest condition to $\rotMatLEq$, namely to the desired attitude, with a displacement due to the parametric uncertainty. Instead, for $\internalTension<0$, is  $\stateSetEqMinus$ the equilibrium point in which the configuration of the load is the closest to the desired one. We can say that these configurations are the most desirable equilibrium configuration of the load in the presence of parametric uncertainties. As stated in \textit{Theorem~\ref{thm:stability}, ${\internalTension>0}$ stabilizes the most desirable equilibrium configuration of the load, which is, instead, unstable if ${\internalTension<0}$.}

\section{The role of the internal forces on the load error caused by parametric uncertainties}\label{sec:error}
In this section, we provide a formal analysis of the role that the internal force plays in determining the load pose at the equilibrium in the presence of parametric uncertainties. We shall consider the simultaneous presence of all the uncertainties listed in Sec \ref{sec:equilibria}. We start considering the load \emph{attitude}.  
\subsection{Load attitude error}\label{ssec:attitude}
	\begin{thm}\label{thm_sens}
	The load attitude error at the equilibrium $\errorRL$, is inversely proportional to the intensity of a positive internal force $\internalTension$. Furthermore, defining 
\begin{align}\label{eq:errR}
&\errorRL =\norm{\rotMatL\vE{1} \times \rotMatLEq\vE{1} }^2,
\end{align}
the error sensitivity w.r.t. $\Delta_m$, $\Delta_b$, $\Delta_{\springCoeff{i}}$, $\Delta_{\length{i}}$, and $\Delta_\ell$, defined as $\dfrac{\partial\errorRL}{\partial\Delta_m}$, $\dfrac{\partial\errorRL}{\partial\Delta_b}$, $\dfrac{\partial\errorRL}{\partial\Delta_{\springCoeff{i}}}$, $\dfrac{\partial\errorRL}{\partial\Delta_{{\length{i}}}}$, and $\dfrac{\partial\errorRL}{\partial\Delta_\ell}$, respectively, is given by:
	  \begin{align} \label{eq:sens_mass}
	  &\dfrac{\partial\errorRL}{\partial\Delta_m} = \frac{-2{\anchorLengthU{1}}\hat{\ell}g^2\alpha\cos{\pitch}^2}{\internalTension^2 L^2} \\ 
	  &\dfrac{\partial\errorRL}{\partial\Delta_b} = \frac{-2{\massLU}\hat{\ell}g^2\alpha\cos{\pitch}^2}{\internalTension^2 L^2}\\
	 &\dfrac{\partial\errorRL}{\partial{\Delta_{ki}}}= \dfrac{\partial\errorRL}{\partial{\Delta_{l0i}}}= 0\\
	   &\dfrac{\partial\errorRL}{\partial\Delta_\ell} = \frac{-2{\anchorLengthU{1}}{\massLU}g^2\alpha\cos{\pitch}^2}{\internalTension^2 L^2}
	  \end{align}
	  where $$\alpha:= (\anchorLength{1} - \Delta b) (\massL-\Delta m) (\ell -\Delta_{\ell}) - \massL\anchorLength{1}.$$
\end{thm}
\begin{proof}
With a positive internal force,  at the equilibrium, \eqref{eq:inc_mass_yaw} with $k=0$ holds. Thus, the quantity $|\pitch^{eq}-\bar{\pitch}|$ is a viable indicator of the attitude error. According to  \eqref{eq:inc_mass_pitch} and for monotonicity of the tan() function, the difference between $\pitchDes$ and $\pitchEq$ varies with the quantity $\condZero g/(L\internalTension\cos{\pitchDes})$, which is inversely proportional to $\internalTension$. Hence, also the error is.
 Rewrite now \eqref{Rd3} in $\frameW$ as:
	 \begin{equation}
	 \rotMatLEquilib\vE{1}\times\left[\left(\anchorLength{1}\massL - \frac{\anchorLengthU{1}{\massLU}L}{\hat{L}}\right)g\vE{3} + L\internalTension\rotMatLEq\vE{1}\right] =\vect{0}.\label{Rd3W}
	 \end{equation}
Define also:
\begin{align}
 &\frac{\anchorLengthU{1}\massLU L-\anchorLength{1}\massL\hat{L}}{\internalTension L\hat{L}}(\rotMatLEquilib\vE{1}\times g\vE{3})
 := \vect{x} .
 \end{align}
 Thus, from \eqref{Rd3W}, we have that $
 \rotMatLEquilib \vE{1} \times \rotMatLEq \vE{1} = \vect{x}$
	 and, from \eqref{eq:errR}, that
	 $
	 \errorRL = \vect{x}^\top\vect{x}.
	$
Regarding the sensitivity, we show the proof for \eqref{eq:sens_mass} only, because the other cases follow the exactly same analysis.
We can write the sensitivity as:
\begin{align}
&\dfrac{\partial\errorRL}{\partial\Delta_m} = 2\vect{x}^\top\dfrac{\partial \vX}{\partial \Delta_m} =\nonumber\\
&= 2[\frac{1}{\internalTension L} \rotMatL\vE{1} \times (\alpha)g\vE{3}]^\top[\frac{1}{\internalTension L} \rotMatL\vE{1} \times(\Delta b -\anchorLength{1})g\vE{3}] \label{eq:sensdm_partial}
\end{align}
Eventually, \eqref{eq:sensdm_partial} can be rewritten as \eqref{eq:sens_mass} by remembering that, given three vectors $\vect{a}, \vect{b},$ and $\vect{c}$  $$(\vect{a} \times \vect{b})^\top (\vect{a} \times \vect{c}) = |\vect{a}|^2 (\vect{b}^\top \vect{c})-(\vect{a}^\top \vect{b})(\vect{a}^\top \vect{c}).$$ 
 Note that we are considering $\internalTension\neq0$ by assumption. \end{proof}
 \begin{rmk}
The definition in  \eqref{eq:errR} is a suitable metric for the attitude error.  if we consider the equilibrium point $\stateSetEqPlus$, namely the one in which the displacement between $\rotMatLEquilib\vE{1}$ and $\rotMatLEq\vE{1}$ is the smallest, and hence our desired equilibrium point.
Firstly, $\rotMatL\vE{1}$ is enough to describe the entire attitude of the \emph{beam-like} load. Secondly, $\errorRL$ is zero when $\rotMatLEquilib = \rotMatLEq$ and increases with the displacement between the two vectors $\rotMatLEquilib\vE{1}$ and $\rotMatLEq\vE{1}$, at least locally (for displacements smaller than $\pm\pi/2$).
 \end{rmk}
 
Moreover, Theorem~\ref{thm_sens} shows that increasing the  intensity of the internal force $\internalTension$ not only makes the attitude error smaller in presence of parametric uncertainties, but it also makes the error more \emph{robust to variations of such uncertainties.} 

This last aspect may be of particular practical interest:  as a matter of fact, parametric uncertainty variations take place every time the actual physical parameters of the system change. A possible real-world scenario is the transportation of objects that are slightly different from each other, e.g., in mass and length. One may want to transport the objects without changing every time the controller parameters for the sake of time, thus dealing with varying parametric uncertainties. Especially interesting, as also highlighted in~\cite{aghdam2016cooperative}, is the variation affecting the CoM position, which may change online when transporting moving masses, i.e. containers of liquids, or boxes with smaller objects free to move inside.
The previous analysis suggests that in all these cases having a larger value of $\internalTension$ is of uttermost benefit, resulting in an error less sensitive to the aforementioned parametric variations.

\subsection{Load position error} \label{sec:err-position}
Differently from what happens to the load attitude error, the load position error at the equilibrium does not necessarily decrease when $\internalTension$ increases. 
While to claim a positive statement a comprehensive proof is needed, as we did in Sec~\ref{ssec:attitude}, to deny a positive statement, as we do in this section,  a  counterexample is enough.
First, it is easy to see that, when only $\length{1}$ is uncertain, the load position error at the equilibrium, $\errorPL:=\pLEquilib- \pLEq$ is 
\begin{equation}
    \errorPL= \Delta_l\frac{b_2mg\vE{3}+\internalTension\rotMatLEq\vE{1}}{\norm{b_2mg\vE{3}+\internalTension\rotMatLEq\vE{1}}}\label{eq:epl_l_inc}.
\end{equation}
Eq. \eqref{eq:epl_l_inc} suggests that $\errorPL$ is equal to a unit vector multiplied by $\Delta_{\length{1}}$, thus, its module is independent of the value of $\internalTension$. Moreover, in the next section, we provide two numerical examples showing that, depending on the specific combination and values of uncertainties, $\errorPL$ may even have a non-monotonic evolution for increasing values of $\internalTension$, with an initial increase or decay. 

We show, however, that the load position error at the equilibrium can be corrected, ideally to zero, without altering the leader-follower architecture nor requiring direct communication between the robots.

We recall that, due to parametric uncertainties, the reference position given to the leader robot is \begin{equation}
\pREqInc{1} = \pLEq +\rotMatLEq {\anchorPosLU{1}} + \left(\frac{{\norm{\cableForceEqInc{1}}}}{{\springCoeffU{1}}} +{\lengthU{1}}\right)\frac{{\cableForceEqInc{1}}}{\norm{{\cableForceEqInc{1}}}}. \label{pr1_expl}
\end{equation}
By using kinematics, \eqref{pr1d3}, and \eqref{pr1_expl}, the load position at the equilibrium is
\begin{align}&\pLEquilib = \pREquilib{1} - \left(\frac{{\norm{\cableForceEquilib{1}}}}{{\springCoeff{1}}} + {\length{1}}\right)\frac{{\cableForceEquilib{1}}}{\norm{{\cableForceEquilib{1}}}} - \rotMatLEquilib \anchorPos{1}\nonumber\\ 
& = \pLEq +\rotMatLEq {\anchorPosLU{1}} + \left(\frac{{\norm{\cableForceEqInc{1}}}}{{\springCoeffU{1}}} + {\lengthU{1}}\right)\frac{{\cableForceEqInc{1}}}{\norm{{\cableForceEqInc{1}}}} - K_A^{-1}\Delta_m g\vE{3} +\nonumber \\ &- \left(\frac{{\norm{\cableForceEquilib{1}}}}{{\springCoeff{1}}} + {\length{1}}\right)\frac{{\cableForceEquilib{1}}}{\norm{{\cableForceEquilib{1}}}} - \rotMatLEquilib \anchorPos{1}. \label{eq:pL}
\end{align} 
From \eqref{eq:pL}, we have an expression for $\errorPL$.
Now, \emph{if the leader robot knows the load position}, it can recognize that, at steady state, $\tilde{\pL}\neq0$ holds, and it can adjust its position reference to $^2\pREqInc{1}$ accordingly, with 
\begin{equation}
^2\pREqInc{1} = \pREqInc{1} - \errorPL.\label{eq:pos_p1_new}
\end{equation}
In this way, there will be a new equilibrium in which the leader robot position is
\begin{equation}
\pREquilib{1}=    {^2\pREqInc{1}} - K_A^{-1}\Delta_m g\vE{3}
\end{equation}
and thus \eqref{eq:pL} becomes 
\begin{equation}
\pLEquilib = \pLEq
\end{equation} 
It is important to highlight that the leader robot position only influences the load position and  not the attitude at the equilibrium, which depends only upon the reference forces computed based on $\rotMatLEq$. Indeed,  by evaluating \eqref{eqn:closedLoopDynamics} at the equilibrium, the last three rows are
\begin{align*}
	\skew{\anchorPosL{1}}{\rotMatLEquilib}^\top\cableForceEquilib{1} +  \skew{\anchorPosL{2}}{\rotMatLEquilib}^\top\cableForceEquilib{2} &= \vZero,
		\end{align*}
		which becomes, substituting $\cableForceEquilib{i}$, eq.  \eqref{Rd3}.
	Hence, the leader robot can correct the load position error, while the internal force independently acts decreasing the attitude error. Because the load's position can be steered relying solely on the leader robot, unlike the load's  attitude which is determined by the cooperative actions of both robots, the control approach  can  maintain its distributed nature. However, to correct the load position error, the leader robot must have access to the load position. This implies that the leader robot would require additional sensors, such as cameras to accomplish this task.

\section{Numerical Validation}\label{sec:exp_num}

  \begin{figure}[t]
    \centering
    \includegraphics[width=0.90\columnwidth]{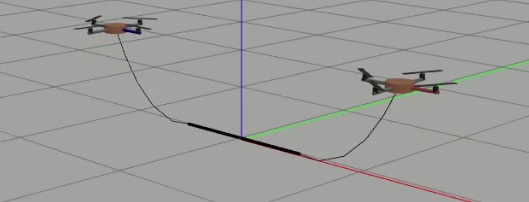}
    \caption{The robots are taking off in the simulated scenario. The sagging effect reproduced by the cable model employed in the simulator is clearly visible.}
    \label{fig:takeoff}
\end{figure}

 \begin{figure*}[t]
\centering
\subfloat[][$\Delta_m\neq0$\label{m_att_err}]{\includegraphics[trim={3cm 8cm 4cm 8cm},clip,width=0.33\textwidth]{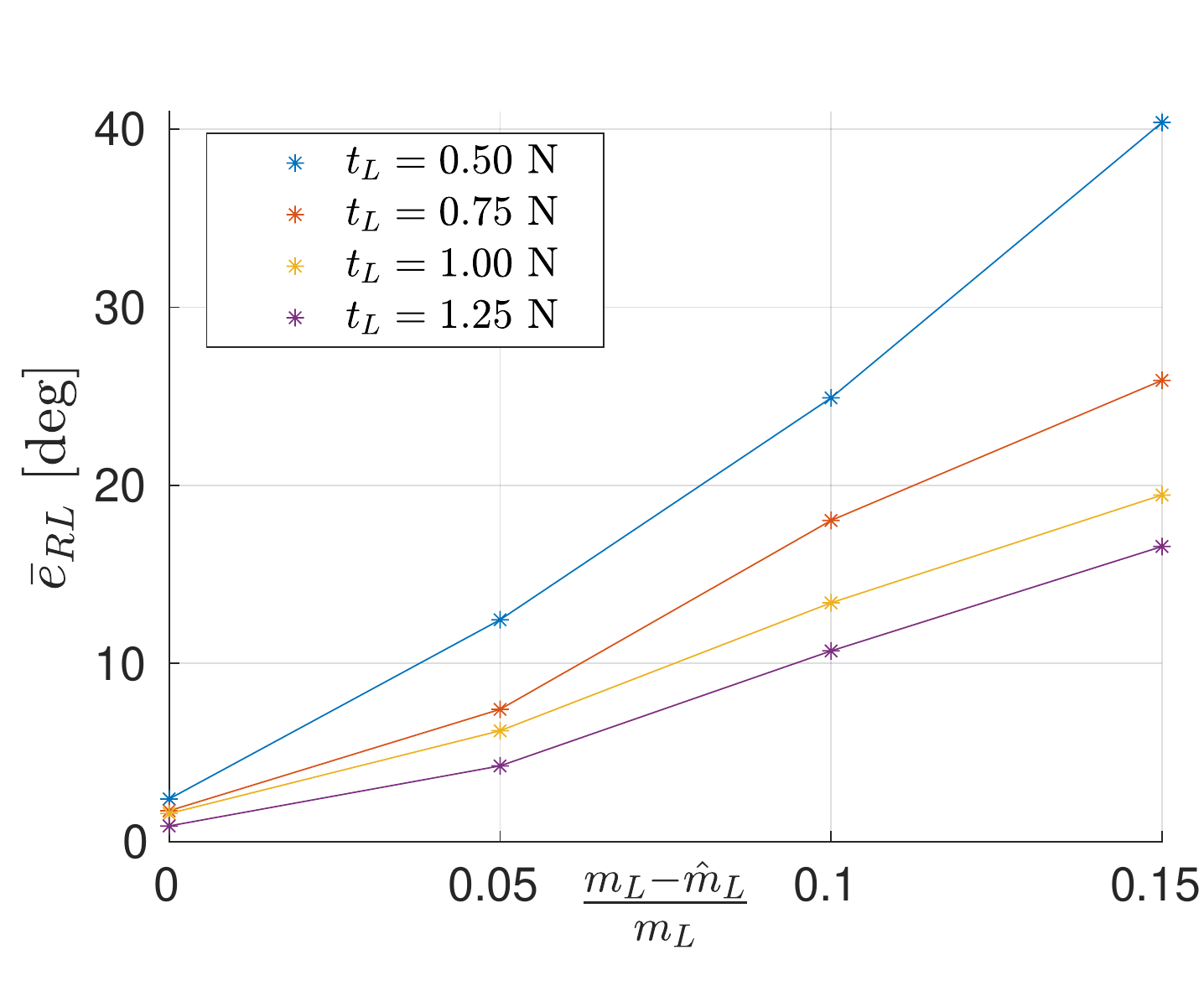}}\subfloat[][$\Delta_\ell\neq0$\label{L_att_err}]{\includegraphics[trim={3cm 8cm 4cm 8cm},clip,width=0.33 \textwidth]{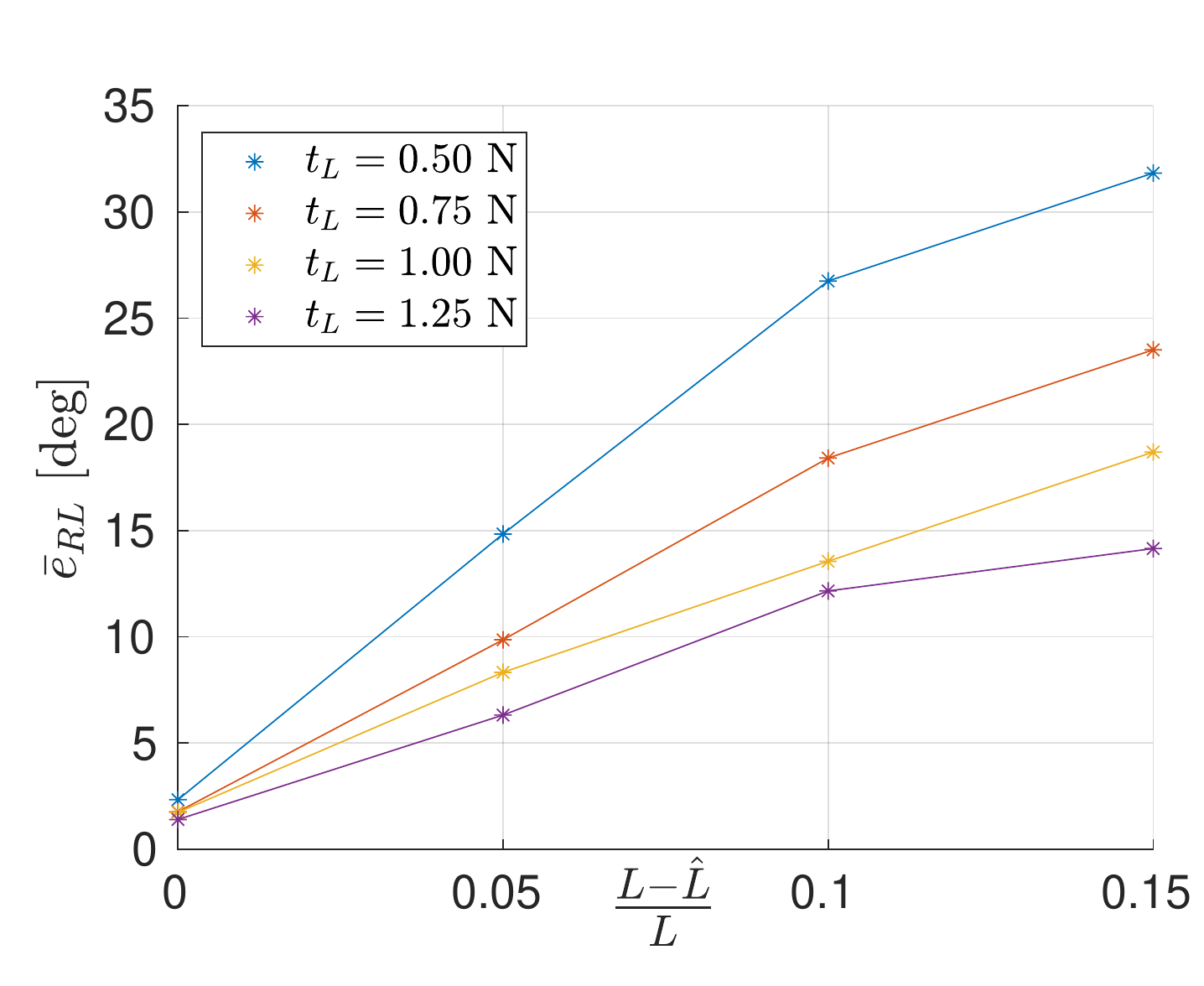}}\subfloat[][$\Delta_{b}\neq0$\label{b1_att_err}]{\includegraphics[trim={3cm 8cm 4cm 8cm},clip,width=0.33 \textwidth]{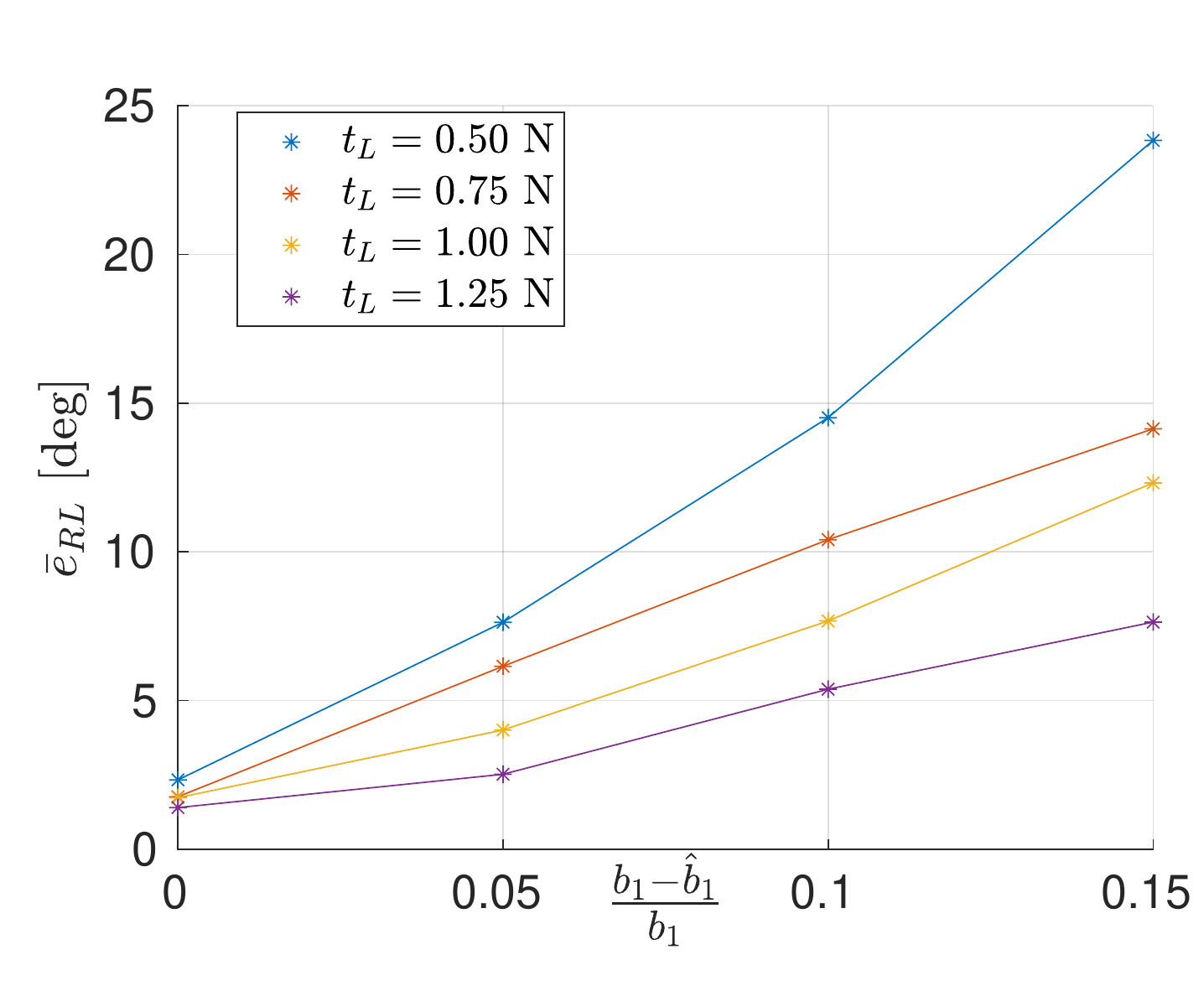}}
\caption{Each point in the plots is a value of the average attitude error at steady state in a simulation with $\internalTension$ as  indicated in the legend, and parametric uncertainty as indicated in the x-axis of the corresponding plot. A total of 40 simulation results are known in these plots.}
\label{fig:att_err}
\end{figure*}
  \begin{figure*}[t]
    \centering
    \subfloat[][$\lengthU{1}=1.15\cdot \length{1}$. As expected, $\pLEquilib\neq\pLEq$, but $\yaw^{eq}=\bar{\psi}$ and $\theta^{eq}=\bar\theta$.  \label{fig:l115_tl1}]{\includegraphics[trim={0.5cm 7cm 0cm 7cm},clip,width=\columnwidth]{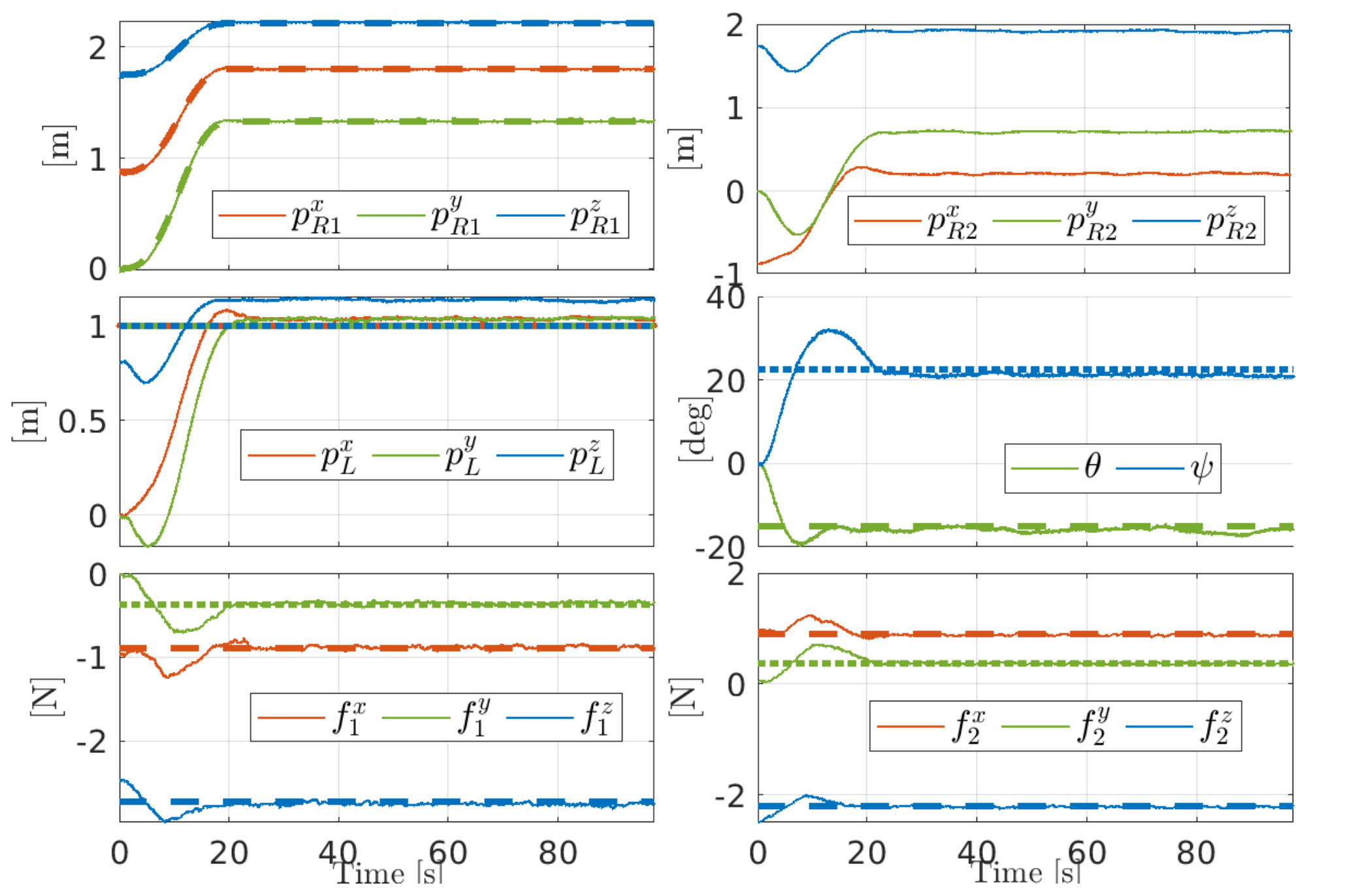}}
    \subfloat[][$\lengthU{2}=1.15\cdot \length{2}$. As expected, $q_L=\bar{q}_L$ at the equilibrium. \label{fig:l215_tl1}]{\includegraphics[trim={0.5cm 7cm 0cm 7cm},clip,width=\columnwidth]{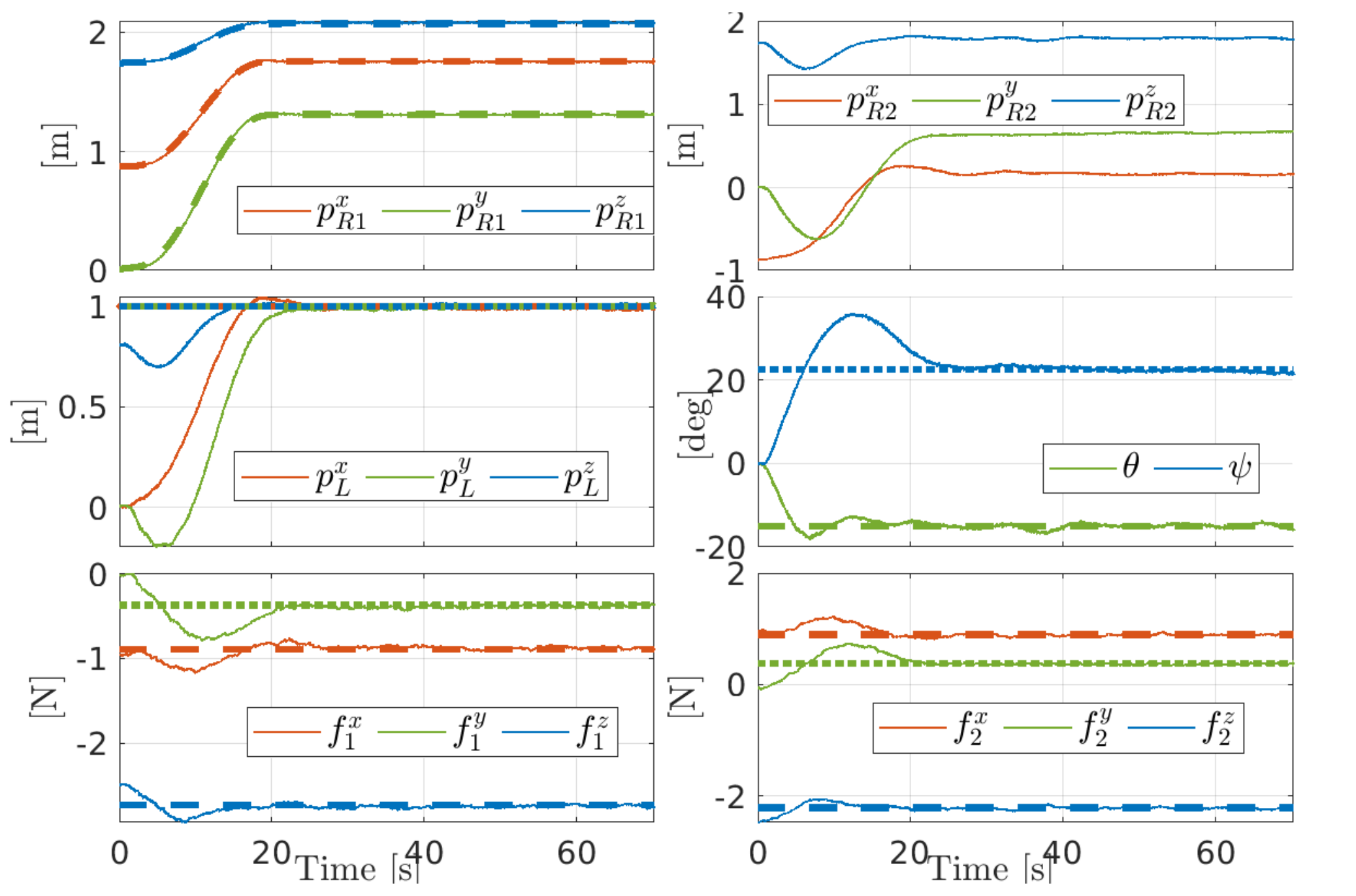}}
    \caption{Simulations for cable parameter uncertainties and $\internalTension=1\rm{N}$. Dotted lines of the same color indicate the corresponding desired quantities.}\label{fig:tineg}
  \end{figure*}
\begin{figure*}[t]
    \centering
  \subfloat[][$\condZero>0$: $\massLU=0.95\massL$,   $\hat{L}=1.05L$. \label{fig:tineg_csipos}]{\includegraphics[trim={0.5cm 6.8cm 0cm 6.8cm},clip,width=\columnwidth]{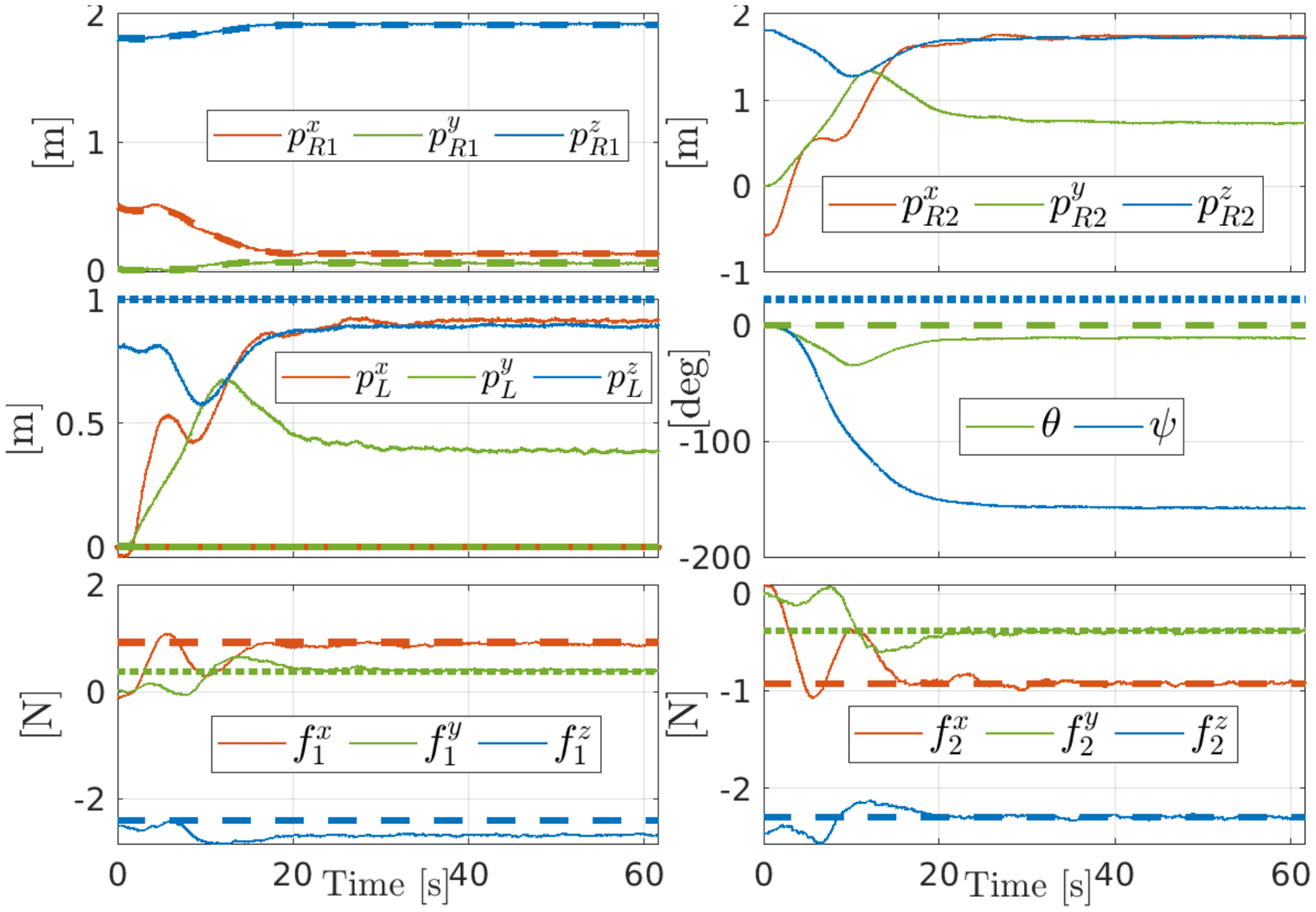}}
    \subfloat[][$\condZero<0$: $\massLU=1.05\massL$,  $\hat{L}=0.95L$ \label{fig:tineg_csineg}]{\includegraphics[trim={0.5cm 6.8cm 0cm 6.8cm},clip,width=\columnwidth]{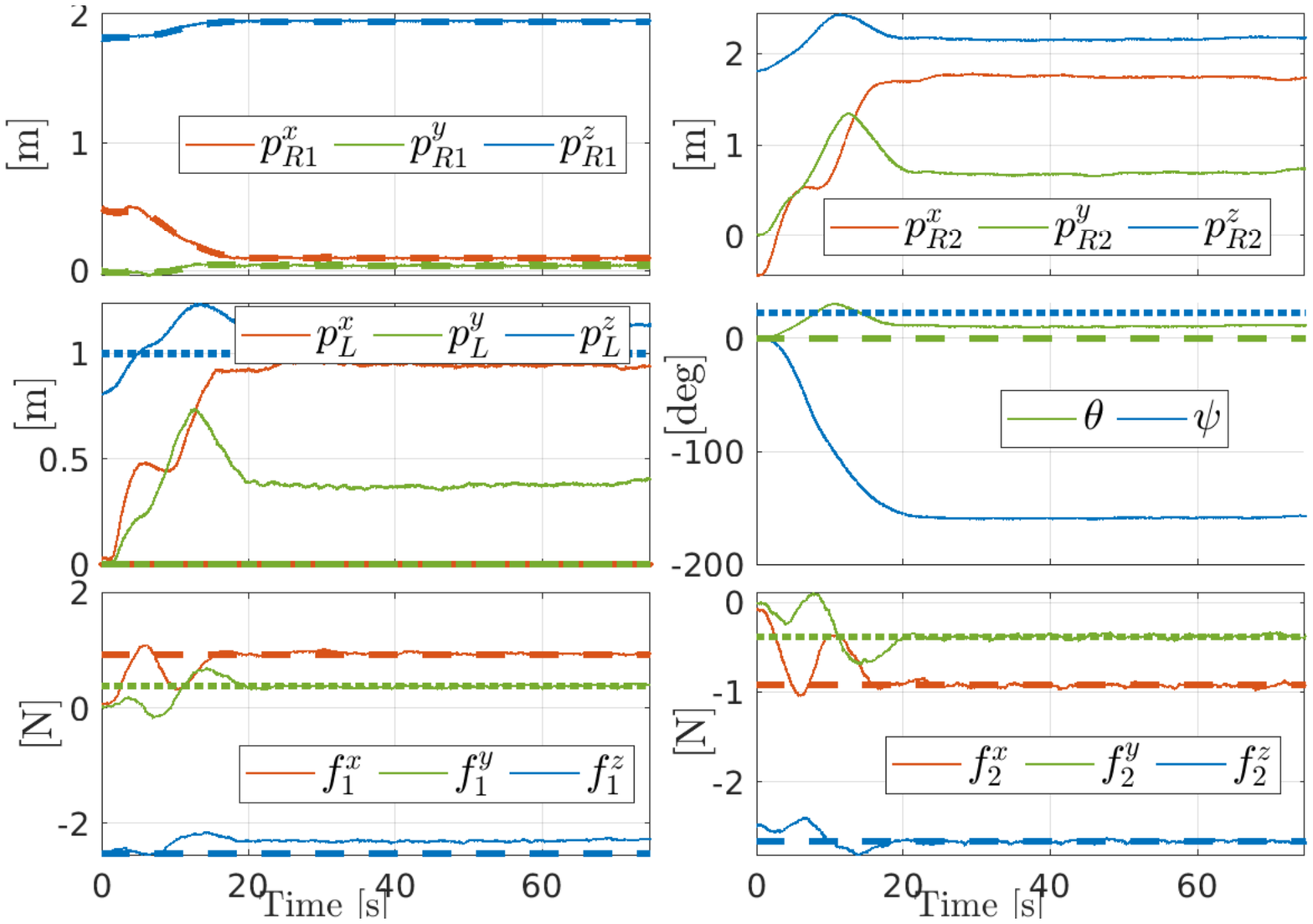}}
    \caption{Simulations for cable parameter uncertainties and $\internalTension=-1\rm{N}$. Dotted lines of the same color indicate the corresponding desired quantities.}\label{fig:tineg_signals}
  \end{figure*}
  %
 Extensive numerical simulations have been carried out  using a URDF description of the system and ODE physics engine in Gazebo. 
  We avoided validating the theoretical results on the same equations used to derive them. The main differences between the  control model used to derive a fully satisfactory theoretical analysis and the  complex simulation model used to study the applicability of the theoretical results in the real world are in the following. 
  \begin{itemize}
      \item Under-actuated quadrotors have been deliberately preferred for the validation since they represent the worst-case in terms of the validity of some of the assumptions made in the theoretical analyses. Validation using fully-actuated aerial robots would have seemed, instead, limiting.
      \item The cables are subject to sagging, which is obtained  by using a series of several links interconnected by passive universal joints, as can be seen in Figure~\ref{fig:takeoff}. 
      \item In the validation, there is no guarantee of perfect trajectory tracking as assumed in the theory but a standard position controller~\cite{2010-LeeLeoMcc} is implemented for each robot. 
      \item The wrench observer proposed in~\cite{2019h-RylMusPieCatAntCacFra} is used to estimate the force applied by the cable on the robot. The observer introduces noisy and delayed measurements when compared to ideal force measurement.
  \end{itemize}
The control software has been implemented in Matlab-Simulink using the Generator of Modules GenoM\footnote{\url{https://git.openrobots.org/projects/genom3}}.  The interface between Matlab and Gazebo is also managed by a Gazebo-genom3 plugin\footnote{\url{https://git.openrobots.org/projects/mrsim-gazebo}}.
All phases of a physical experiment, starting with takeoff, are replicated in the simulated environment using a state machine, 
ensuring that the results are as realistic as possible. After the takeoff, the two robots lift the load, and the admittance controller is activated right after.

The robot models are two quadrotors weighing 1.03 kg and having a maximum thrust for each propeller of 6 N. They are equipped with two light cables of length 1 m and attached to a bar.
The bar is a one-meter-long link with a mass of 0.5 kg.    
 
 \begin{figure*}[t]
    \centering
    \includegraphics[width=\textwidth]{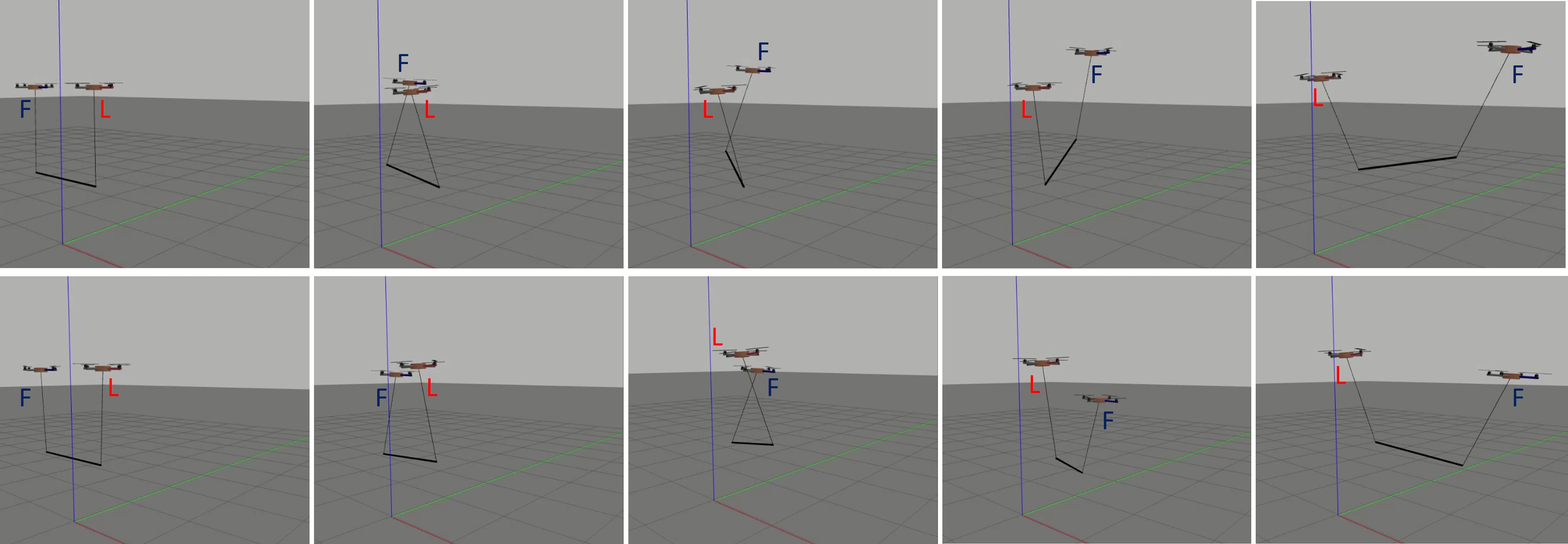}
    \caption{Frames from two simulated scenarios for $\internalTension<0$. Top: $\condZero<0$; bottom: $\condZero>0$. The leader has a red arm and the follower a blue arm. However, to facilitate the distinction among them, a red letter `L' indicated the leader and a blue `F' the follower.}
    \label{fig:tineg}
\end{figure*}

\subsubsection{Case of $\internalTension>0$}\label{subs:tI>0}
Figure~\ref{fig:att_err} contains the average load attitude error at a steady state in a total of 40 simulations. The average is computed in a 2-second time window.  In all those simulations, $\pLEq=[1\ 1\ 1]^\top$ m, the  load  desired yaw is $\bar{\psi}=\frac{\pi}{8}$ radians, and the desired pitch  $\bar{\theta}=-\frac{\pi}{12}$ radians. In each of the three plots in Figure~\ref{fig:att_err}, for four different values of the internal force, ${\internalTension=\{0.5,\ 0.75,\ 1,\ 1.25 \}\rm{N}}$, four different simulation results are displayed for each relative error equal to $0\%, 5\%, 10\%, 15\%$  on a specific uncertain parameter considered separately from the others. Specifically, Figure~\ref{m_att_err}, considers the uncertainty on $\massL$, Figure~\ref{L_att_err} on $L$, and Figure~\ref{b1_att_err} on $\anchorLength{1}$. Even in the absence of uncertainties, small errors of less than 
2.5 degrees in the bar's attitude control can be found.  This can be due to minor tracking errors or possible biases in the wrench observer, which estimations are unbiased as soon as the robot takes off. These considerations ignore the external forces applied by the loose cables at the startup phase.
 From all the three figures one can appreciate the beneficial effect of larger values of $\internalTension$ on the attitude error: for the same value of the uncertainty, the attitude error decreases if the $\internalTension$ increases. Moreover, the plots show that for every value of $\internalTension$, increasing the uncertainty on one parameter increases the attitude error, as expected, but, especially, the increase is smaller for high values of $\internalTension$ (this can be seen by the slope of the lines in the plots). These results confirm the theoretical findings collected in Theorem \ref{thm_sens}. 
 Figure~\ref{fig:l115_tl1} and \ref{fig:l215_tl1} provide  validation of the theoretical results on the effect of the uncertainties affecting the cable parameters. The leader robot position reference is not given as a step, but the robot follows a 5-th order polynomial trajectory to reach the desired position. An error of 15\% is considered to affect the length of the leader and follower robot's cable in Figure~\ref{fig:l115_tl1} and \ref{fig:l215_tl1}, respectively. In both cases, $\internalTension=1 \rm{N}$. Note that the displayed time starts  after the admittance controller activation.  The reader can appreciate how uncertainties on the leader robot's cable model only cause  $\pLEquilib\neq\pLEq$, while $\rotMatLEquilib=\rotMatLEq$, and how the follower robot's cable parameters are not needed to control the load pose, as expected from~\eqref{pLeqdl0} and~\eqref{eq:Rlinc}. 

\subsubsection{Case of $\internalTension<0$}
Figure~\ref{fig:l115_tl1} and \ref{fig:l215_tl1} provide  validation of the theoretical results on the effect of the uncertainties affecting the cable parameters. The leader robot position reference is not given as a step, but the robot follows a 5-th order polynomial trajectory to reach the desired position. An error of 15\% is considered to affect the length of the leader and follower robot's cable in Figure~\ref{fig:l115_tl1} and \ref{fig:l215_tl1}, respectively. In both cases, $\internalTension=1 \rm{N}$. Note that the displayed time starts  after the admittance controller activation.  The reader can appreciate how uncertainties on the leader robot's cable model only cause  $\pLEquilib\neq\pLEq$, while $\rotMatLEquilib=\rotMatLEq$, and how the follower robot's cable parameters are not needed to control the load pose, as expected from~\eqref{pLeqdl0} and~\eqref{eq:Rlinc}. 

\subsubsection{Case of $\internalTension<0$}
Here we show the behavior of the system with $\internalTension<0$. The unstable nature of the desired configuration with no parametric uncertainties was shown in~\cite{2018h-TogGabPalFra}. The simulations of the realistic system, in accordance with Theorem \ref{thm:stability},  show that $\stateSetEqMinus$ is unstable when the parametric uncertainties are considered. We report the results of two simulations with $\pREq{1}=[0\ 0\ 1]^\top$ m, $\bar{\psi}=\frac{\pi}{8}$ radians, and $\bar{\theta}=0$ radians (desired horizontal bar). We simulate an uncertainty of $5\%$ both on $\massL$ and $L$ such that (i) $\condZero>0$ (we chose $\massLU<\massL$ and $\hat{L}>L$) and (ii) $\condZero<0$ (thanks to $\massLU>\massL$ and $\hat{L}<L$). In both cases, we obtained that, in accordance to Theorem \ref{thm:stability}, the system converges to $\stateSetEqPlus$. Since $\internalTension<0$, that means, as reported in Figure \ref{fig:tineg}, that we have $\yaw^{eq} = \bar{\yaw}-\pi$, while $\theta^{eq}$ varies according to the sign of $\condZero$, as expected (see Figure~\ref{fig:subfig2}). The cable forces were observed to be as desired,  except for the vertical component of $\cableForceEq{1}$, as expected due to $\Delta_m\neq0$ according to~\eqref{f1dm}. Figure~\ref{fig:tineg} shows the behavior of the system in the described cases through screenshots of the Gazebo environment, and Figure \ref{fig:tineg_signals} shows the evolution of the main quantities during the simulated tasks.

\subsubsection{Case of $\internalTension=0$}
When it comes to the case in which $\condZero\neq0$ and $\internalTension=0$, clearly, the sole equilibrium configurations are not really  attainable: all elements of the system are supposed to be aligned vertically, one on top of the others (see Figure \ref{fig:equilibrium:ZeroTi}). 
\begin{figure}[t]
    \centering
   \subfloat[][$\condZero>0$]{\includegraphics[width=0.5\columnwidth,height=3.5cm]{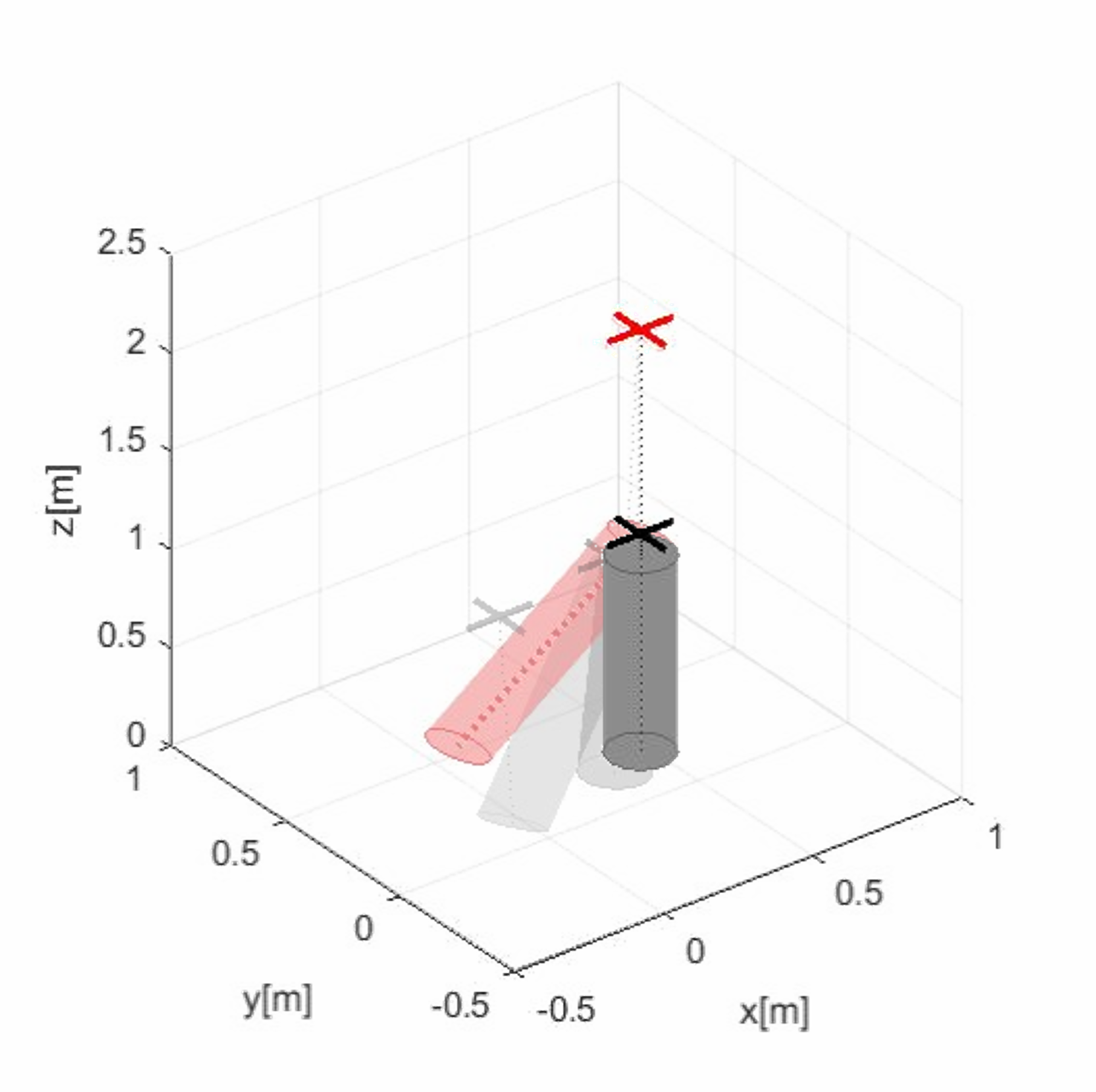}}
    \subfloat[][$\condZero<0$]{\includegraphics[width=0.5\columnwidth,height=3.5cm]{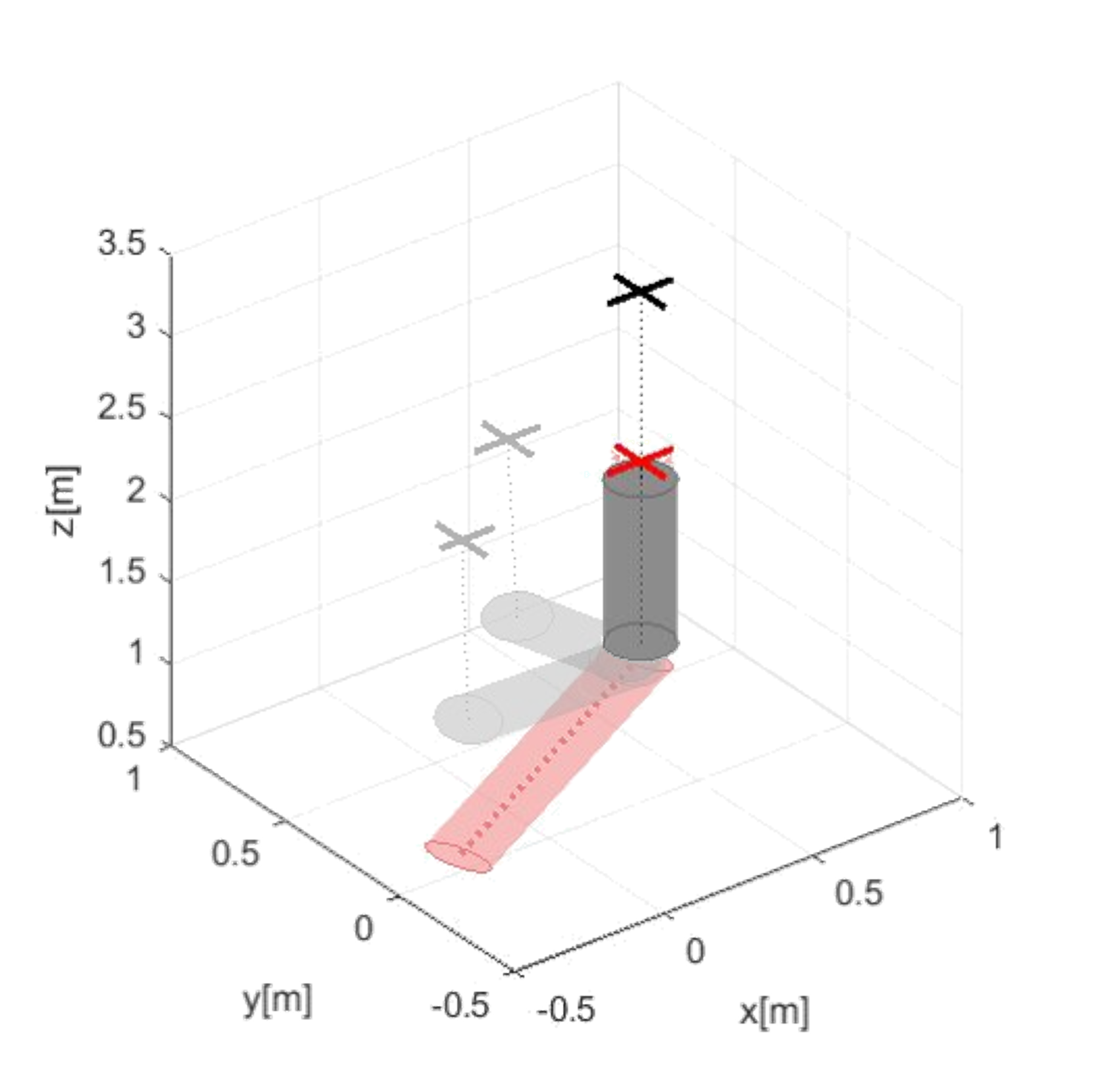}}
    \caption{Superposition of different instants of simulations with $\condZero\neq0$ and $\internalTension=0$. The solid image is the final equilibrium. The grey cylinder is the load, and the red cross is the
leader robot. The red cylinder is the desired (identical to  the initial) pose of the load.}
    \label{fig:ti0unc}
\end{figure} When simulating such condition in Gazebo, we found that numerical issues arise as the system approaches the expected configuration in which the link that models the load and those that model the cables are vertically aligned. Despite the practical irrelevance of the considered case, with the objective of demonstrating the validity of the theoretical results, simulations have been carried out also for this case, using the Matlab-Simulink simulator used in~\cite{2018h-TogGabPalFra}.

In that simulator, the cables are modeled as mass-less extensible elements and the force is directly retrieved by the model of the cable without resorting to a wrench observer. Nevertheless, underactuated quadrotors are still considered, as well as the same trajectory controller. The results of two simulations can be found in Figure \ref{fig:ti0unc}. Even though the load has been initialized in the desired configuration, with position $\pLEq=[0\ 0\ 1 ]^\top$ and the same desired yaw and pitch as before, it moves to the vertical equilibrium, with the leader on top when $\condZero>0$, and the follower on top when $\condZero<0$, as explained by the stability analysis in Sec. \ref{sec:stability}.
\subsubsection{Position Error}
First, we provide in Figure~\ref{fig:epl} two examples of the different behavior of $\errorPL$ when  $\internalTension$ increases and different values of the uncertainties are present. This fully supports the finding that the load position error at the equilibrium does not necessarily decrease when $\internalTension$ is increased.
\begin{figure}[t]
    \centering
    \subfloat[][$\frac{\hat{m}}{m}=0.85$, $\frac{\lengthU{1}}{\length{1}}=0.9$]{\includegraphics[trim={3cm 8cm 4cm 8cm},clip,width=0.48\columnwidth]{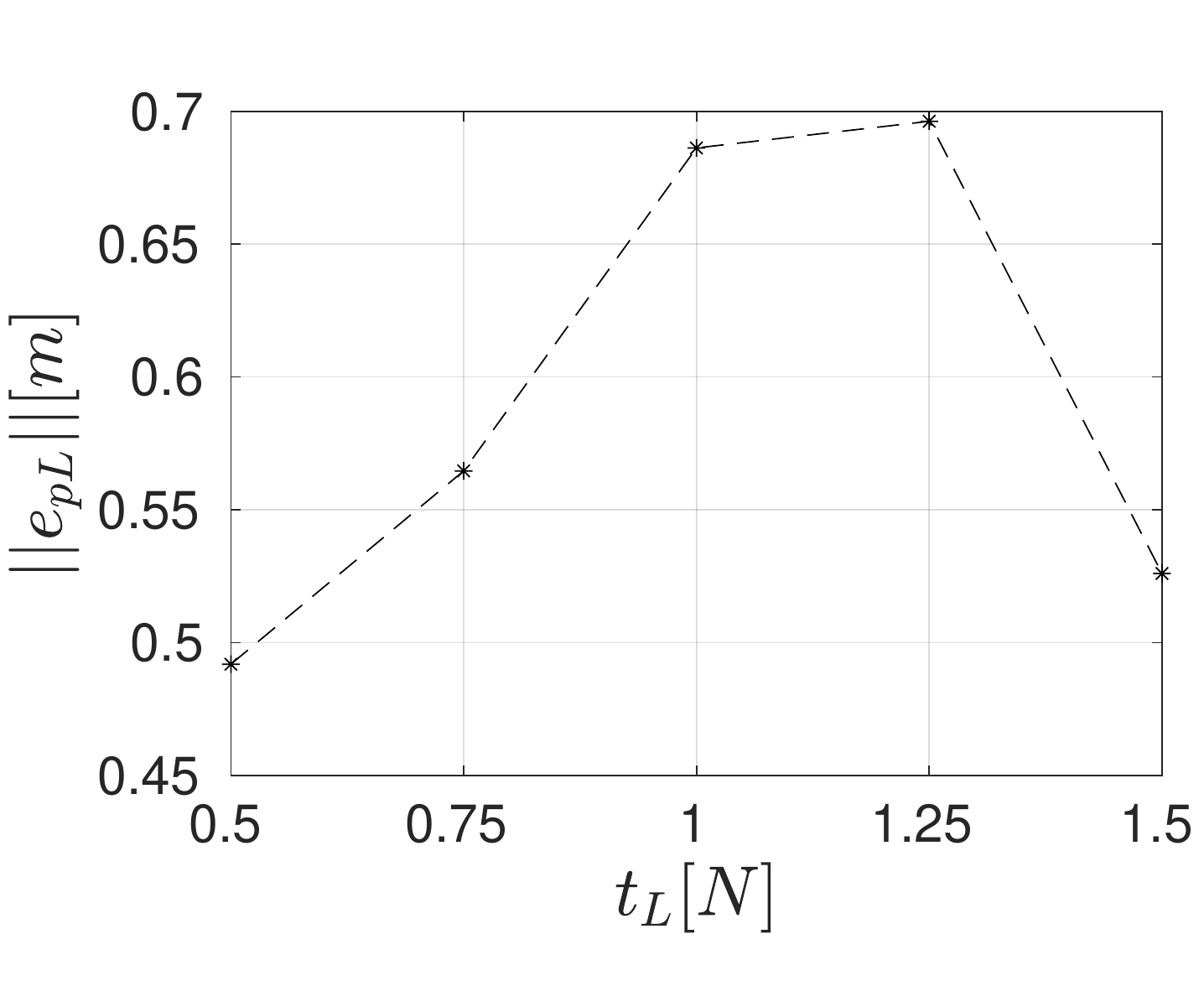}}
    \subfloat[][$\frac{\hat{m}}{m}=0.85$, $\frac{\lengthU{1}}{\length{1}}=1.2$]{\includegraphics[trim={3cm 8cm 4cm 8cm},clip,width=0.48\columnwidth]{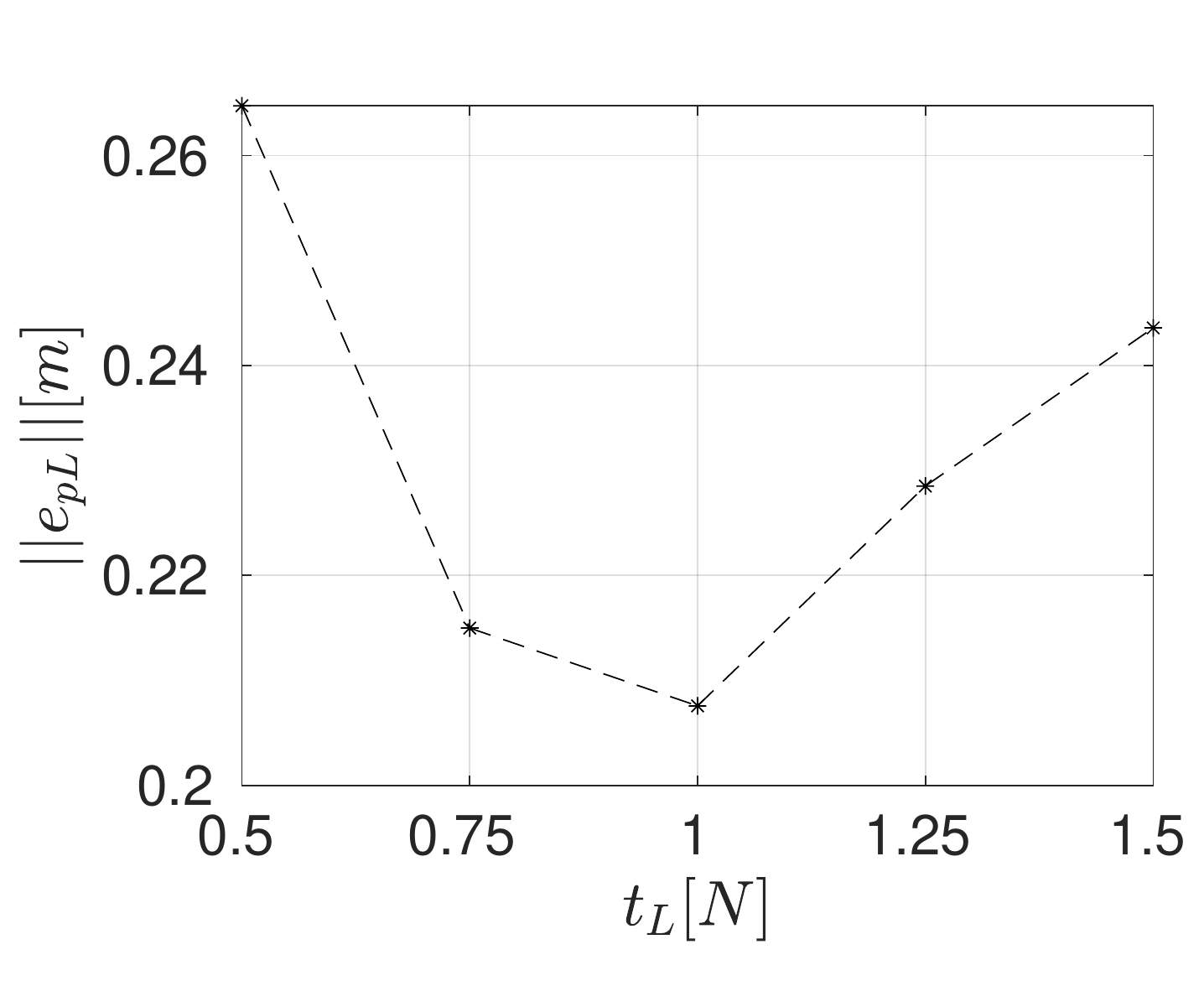}}
    \caption{$\norm{\errorPL}$ for different values of $\internalTension$ in two cases in which different values of the uncertainties are considered on two parameters, $m$ and $\length{1}$. $\errorPL$ does not always decrease when $\internalTension$ increases.}
    \label{fig:epl}
\end{figure}
\begin{figure}[t]
    \centering
    \includegraphics[trim={0.5cm 7cm 0cm 7cm},clip,width=\columnwidth]{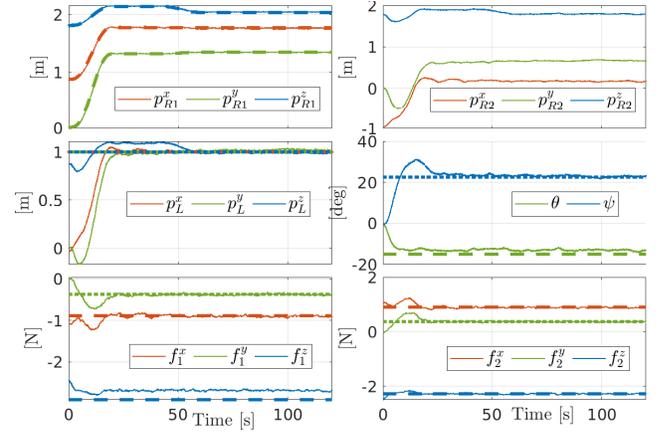}
    \caption{Simulation results for $\internalTension=1$ N and a 5\% error on each parameter. Around Time = 41 s, signed by a red vertical line, the leader robot  reads the load position, and corrects its own reference position in order to zero the load position error. Dotted lines of the same color indicate the corresponding desired quantities.}
    \label{fig:pos_corr}
\end{figure}
Anyway, as we have seen from the theory, it is possible to correct the error of the load position by acting solely on the leader robot reference position. This, in turn, does not affect the regulation of the load attitude. In Figure~\ref{fig:pos_corr}, we report the results of a Gazebo simulation in which the initial and desired load pose are as in Sec. \ref{subs:tI>0}, $\internalTension=1$ N, and an error equal to 5\% of the nominal value is considered on each uncertain parameter. After 41 s, the leader robot corrects its reference position based on the position of the load according to~\eqref{eq:pos_p1_new}. The results show that, consequently, the load is steered to the desired position when the new equilibrium is reached. On the other hand, as expected from the theory (see Eq. \eqref{Rd3} and \eqref{f1d3}),  due to the inaccurate knowledge of the system parameters, the value of the pitch angle and the leader robot's cable force at the equilibrium do not match the desired values. Also, as predicted, one can observe in Figure~\ref{fig:pos_corr} that their values are not affected by the change in the leader robot position.

\section{Experimental Validation}\label{sec:exp_real}
\subsection{Experimental Setup}
\subsubsection{Hardware}
\newcommand{\Appendix}{Appendix~}
The system is made of a 2-meter-long carbon fiber bar carried by two UAVs by means of two cables that connect the robots at the bar's end. Each cable is 1~$\rm m$ long, the bar weighs 0.300~$\rm kg$ and each UAV weighs 1.03 $\rm kg$. The cable anchoring points are installed on the robots' underside at a distance $\displacement = [0~0~-d]^\top$ from their CoM, where $d=$0.15 $\rm cm$.
Such a geometrical configuration changes the process by which the leader reference position $\pREqInc{1}$ is generated as it is explained in the Appendix. 
In addition, the aerial vehicles have an onboard PC, four ESCs (Electronic Speed Controllers) that control the propeller speed in closed-loop~\cite{2017c-FraMal}, and a flight controller~\cite{2017c-FraMal}. 
\subsubsection{Software}
The control architecture runs in part onboard and in part on a desktop PC. A state-of-the-art UKF-based state estimation, which fuses Motion Capture measurements at 120 $\rm Hz$ with the IMU measurements at 1 $\rm kHz$, and a geometric control are carried out as part of the onboard task at 1 \rm kHz. The admittance filter and wrench observer are implemented in  Matlab/Simulink and run on the desktop PC. Wi-fi is used for command and data transfer between the desktop PC and the onboard computers at 100 $\rm Hz$.

A picture taken from the experiments  and highlighting the main setup components is in  Fig~\ref{fig:exp}. 
\begin{figure}[t]
    \centering
    \includegraphics[width=0.9\columnwidth]{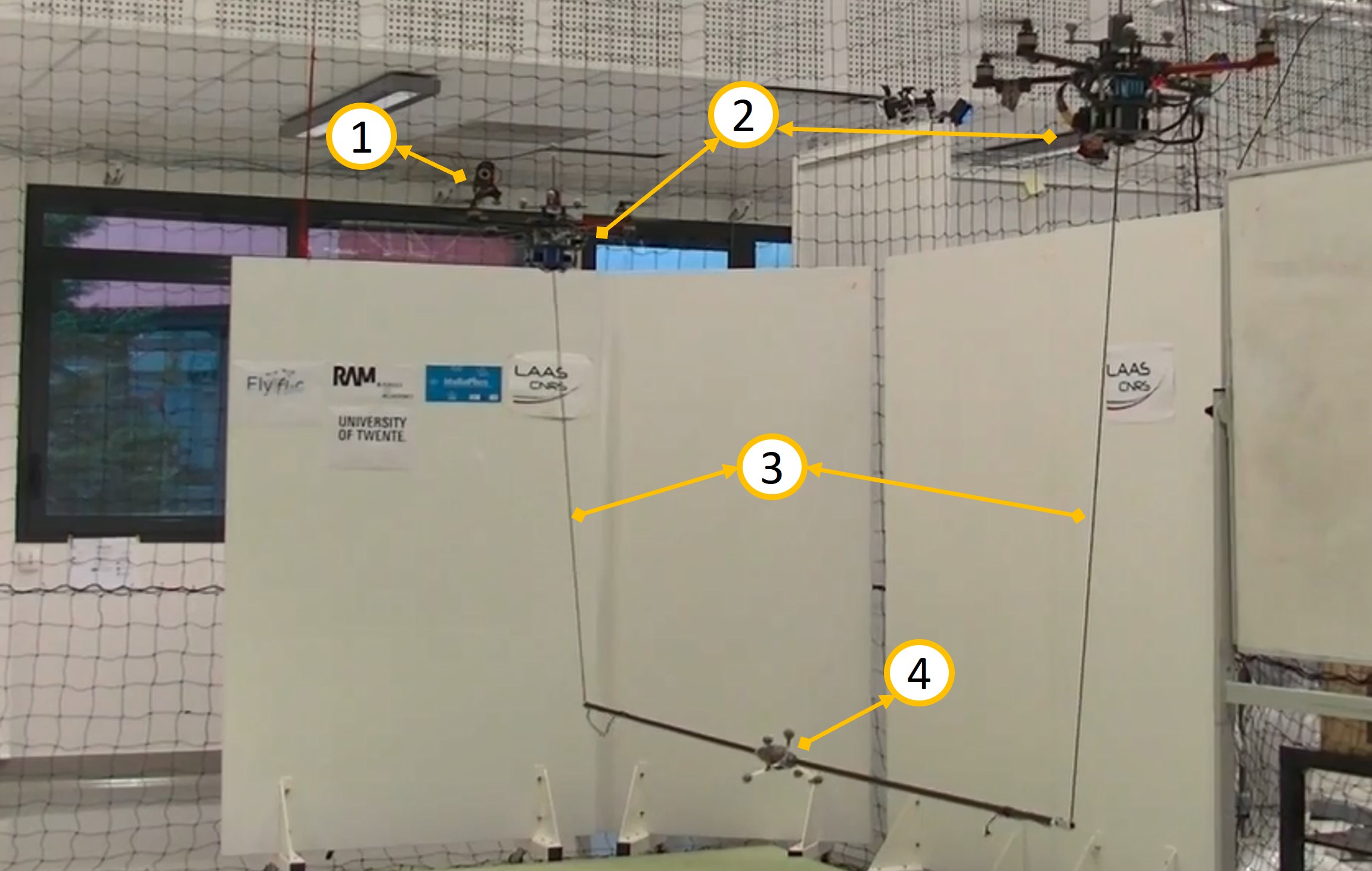}
    \caption{Picture from the experiment showing: one of the motion capture system cameras 1); two quadrotors 2); two cables 3) connecting the robots to the carbon fiber load 4). The markers on the bar track the pose of the object for validation purposes.}
    \label{fig:exp}
\end{figure}

\subsubsection{Experimental Results}
Two main sets of experiments were carried out: one in which the controllers use as accurate as possible values of the system parameters; one in which the controllers use values of the parameters that differ by 10\% from the accurate corresponding value. 
We refer to the former case as `without uncertainty', and to the latter as `with uncertainty'.

For each case, we performed three tests in which the same manipulation task is carried out for three values of $\internalTension$, equal to~0\,N, 1.5\,N, and 3\,N. The task execution, as the simulated one, starts with  initial steps in which the load, from position $\pL(0)=[0\ 0\ 0]^\top$\,m and zero yaw and pitch angles, is lifted by the robots through simple upwards motions; hence, the proposed controller is activated and the robots try to bring the load to $\pLEq=[0.5\ 0\ 1.5]^\top$\,m with $\yawDes=11.5$\,deg and $\pitchDes=-6.9$\,deg.

The case with no uncertainties is depicted in Figure \ref{fig:exp_att_no_unc} for all three experiments. The evolution of the attitude error of the load is displayed, in the form of quantities $\pitch-\pitchDes$ and $\yaw-\yawDes$. As expected from \eqref{eq:inc_mass_yaw}, the yaw angle at the equilibrium coincides with the desired value. Instead, the pitch angle converges to an arbitrary value when $\internalTension=0$, in this case with an error around 18 deg, when $\internalTension=0$. When a positive internal force is applied, the attitude error is reduced up to about 2.5 deg.
\begin{figure}[t]
    \centering
    \includegraphics[trim={3cm 8cm 4cm 9cm},clip, width=0.68\columnwidth]{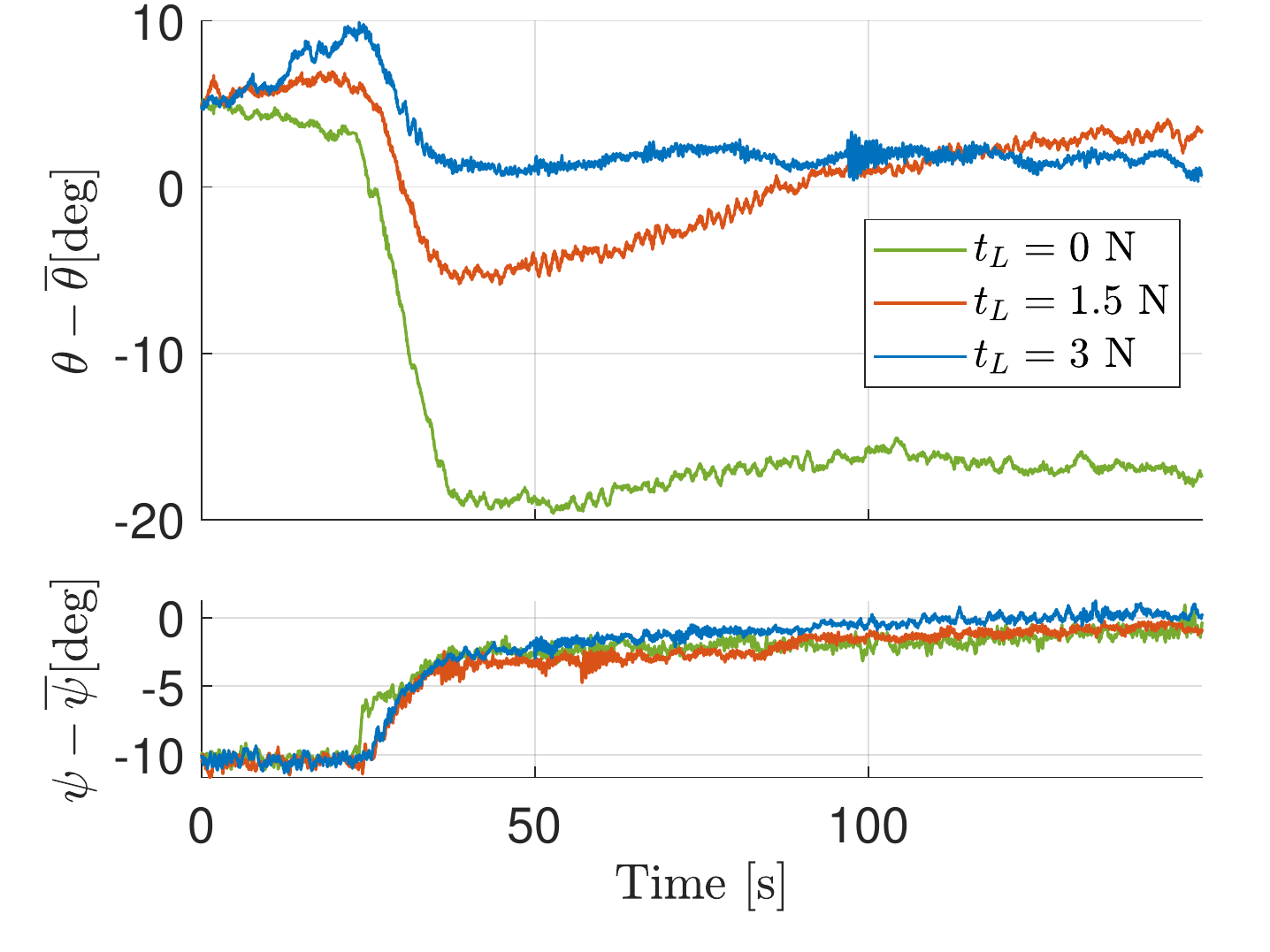}\includegraphics[trim={7cm 7.5cm 7cm 9cm},clip, width=0.32\columnwidth]{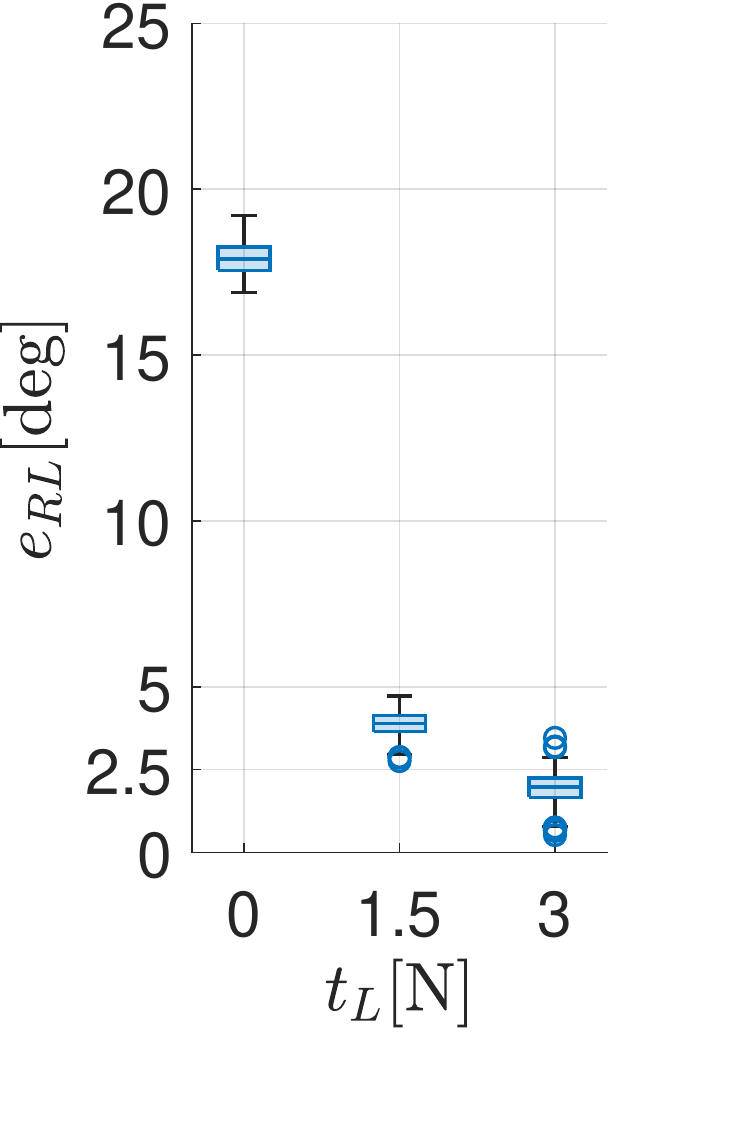}
    \caption{Evolution of the pitch (top) and yaw (bottom) angle errors during three experiments  with no parameter uncertainties for different values of $\internalTension$. As expected, the yaw angle converges to the desired value while for the pitch to do the same, $\internalTension>0$ is needed. On the left is the boxplot of the average attitude error over the last 20 seconds of the experiments.}
    \label{fig:exp_att_no_unc}
\end{figure}
\begin{figure}[t]
    \centering
    \includegraphics[trim={3cm 8cm 4cm 9cm},clip, width=0.68\columnwidth]{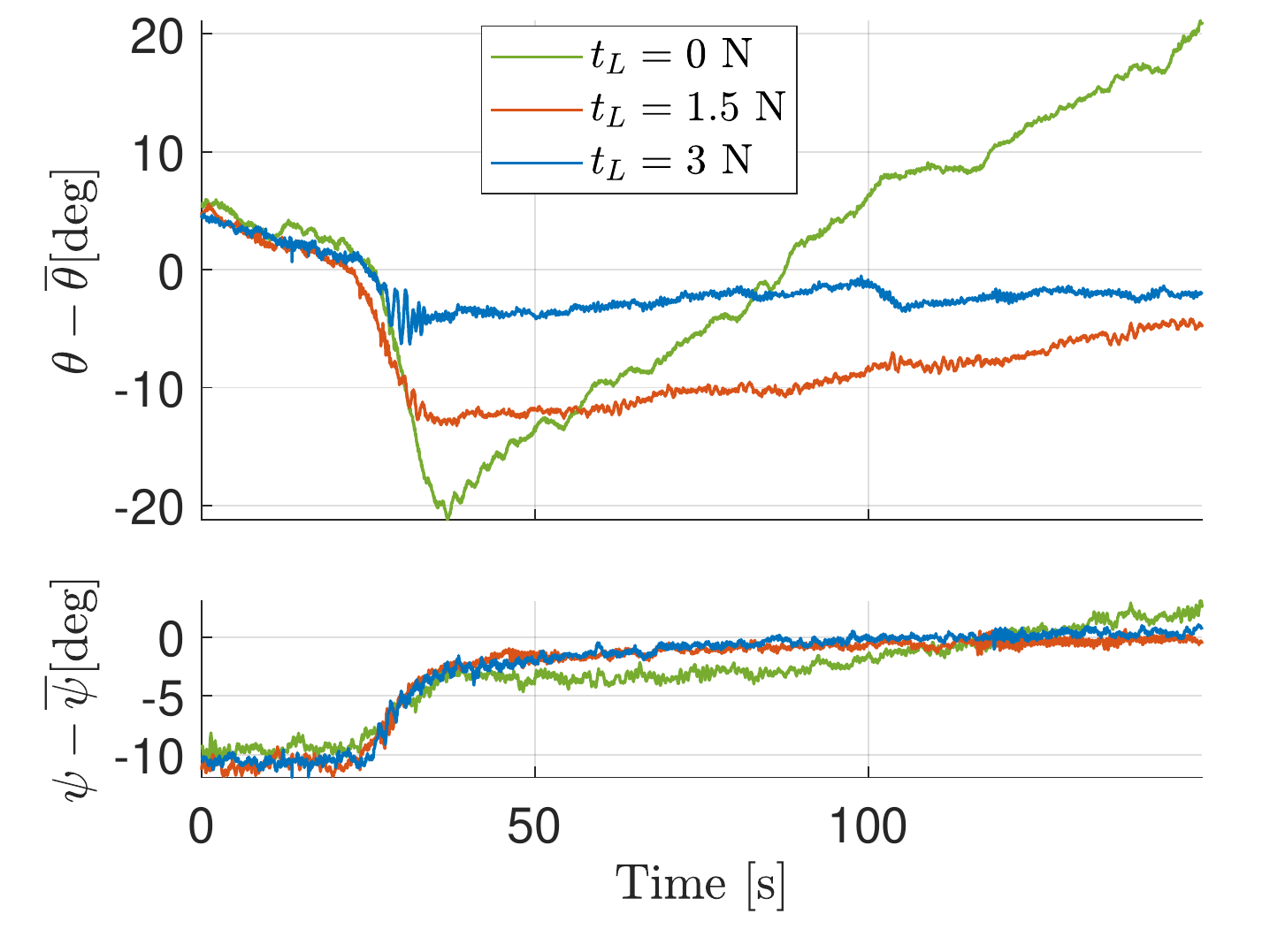}\includegraphics[trim={7cm 7.5cm 7cm 9cm},clip,width=0.32\columnwidth]{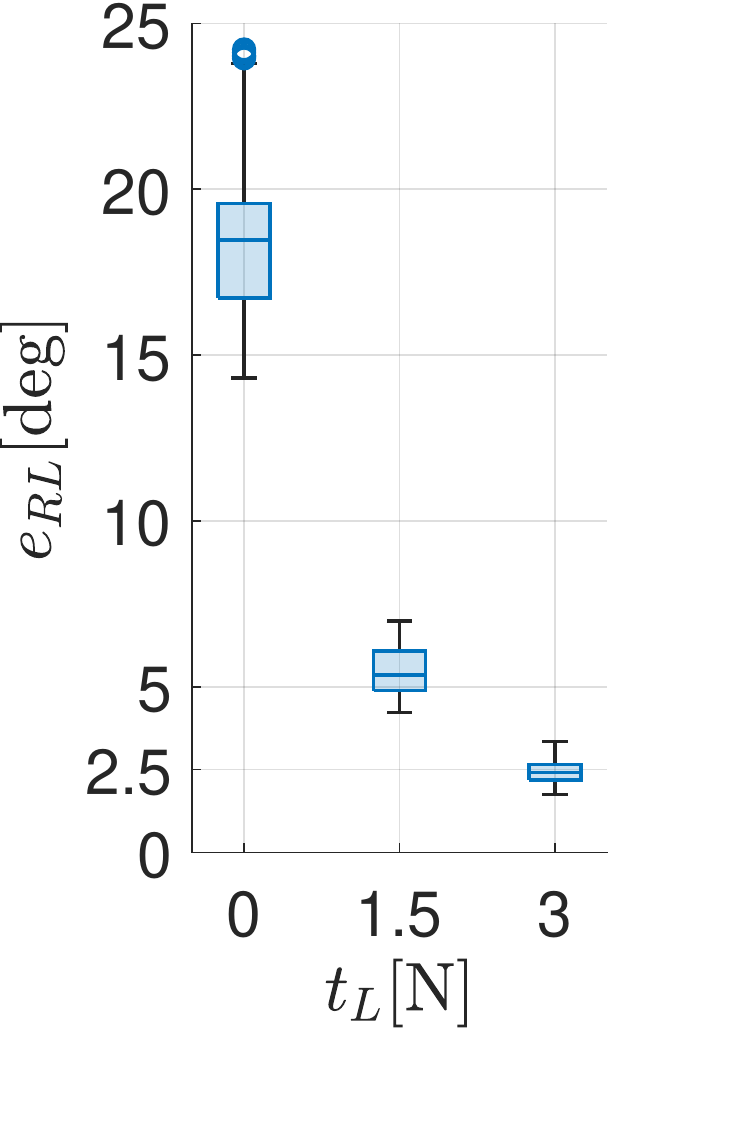}
    \caption{Evolution of the pitch (top) and yaw (bottom) angle errors during three experiments  with 10\% uncertainty on each parameter for different values of $\internalTension$. As expected, the yaw angle converges to the desired value while the pitch angle increases for $\internalTension=0$ as the load becomes more and more vertical. For $\internalTension>0$, the pitch becomes closer and closer to the desired value. On the left is the boxplot of the average attitude error over the last 20 seconds of the experiments.}
    \label{fig:exp_att_unc}
\end{figure}
\begin{figure}[t]
    \centering
    \includegraphics[trim={0.5cm 7cm 0cm 7cm},clip,width=\columnwidth]{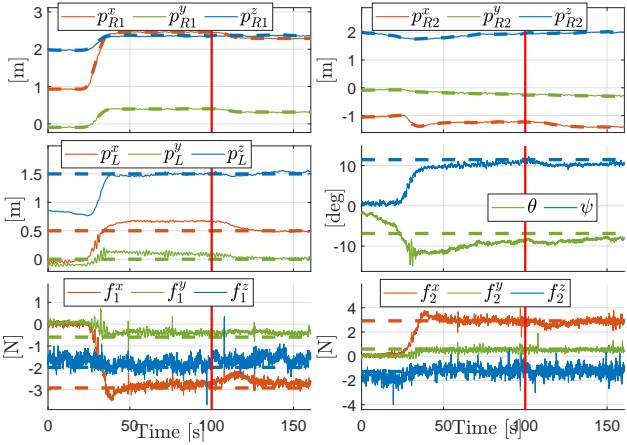}
    \caption{Experimental results: $\internalTension=3~\rm{N}$ and 10\% error on each parameter. At Time=100 s, signed by a vertical red stripe, the leader robot corrects its reference position to zero the load position error.}
    \label{fig:pos_corr_exp}
\end{figure}

Figure \ref{fig:exp_att_unc} shows the three tests in the case with uncertainties. Again, in accordance with \eqref{eq:inc_mass_yaw}, there is no error on the load yaw angle at the equilibrium. The equilibrium configuration could not be reached with $\internalTension=0$. This is also the case in general since so would imply that all bodies are vertically aligned. However, the reader can clearly appreciate the increasing pitch angle evolution, in line with \eqref{Rd3}. As soon as $\internalTension>0$, the pitch angle at the equilibrium approaches the desired value, as expected by \eqref{eq:inc_mass_pitch}. Box plots of the equilibrium error between 130\,s and 150\,s of the task execution are displayed in Figure \ref{fig:exp_att_no_unc} and \ref{fig:exp_att_unc} for both cases, i.e. with and without  uncertainties, respectively.

The evolution of all the most relevant quantities  can be appreciated in Figure \ref{fig:pos_corr_exp} for another task execution with uncertainties and $\internalTension=3$\,N. Furthermore, in that experiment, the leader robot corrects its own reference position at  Time=100 s according to \eqref{eq:pos_p1_new} and, as a consequence,  the load equilibrium position is also adjusted. This validates the load position correction method involving solely the leader robot. Videos from the experimental validation can be found in the multimedia attachment.
\section{Conclusions}\label{sec:conclusions}
This work concerns the decentralized cooperative manipulation of a cable-suspended load by two aerial
robots in the absence of direct communication. The robots are controlled with a leader-follower  scheme  achieved through an admittance controller on each robot. The controllers make use of nominal system parameters that are subject to uncertainty. 
The equilibrium points and their stability are formally studied. The theory demonstrates how an  internal force that tends to stretch the load longitudinally, generated by non-vertically operated cables, is beneficial in terms of stability of the load pose control as well as in terms of robustness to the uncertainties. The complete theoretical results are validated through numerical simulations embedding additional realistic effects.

In the future, an extension to non-beam rigid loads with uncertain parameters will be formally addressed. Note that in the case of generic rigid objects,  $N>2$ robots will be considered as two robots would not be able to control the full pose of the cable-suspended object: rotations around the line connecting the two cables attaching points on the object would not be controlled \cite{MicFinKum2011}. The manipulation of deformable objects in a communication-less setup is an interesting extension. Experimental tests outdoors could be valuable to assess the robustness of the method in windy conditions and when relying on outdoor state-estimation techniques. Investigating how the approach could benefit from the introduction of limited communication between the robots, e.g., low-frequency communication, is an interesting future direction. Exploring the possibility of  communication-less trajectory tracking and of adaptive laws is left as future work. 

\appendix \label{appendix}
Underneath the robots, at a distance $\displacement$ from the CoM ($\displacement~=~[0~0~0.15]^\top \rm cm$), the cable anchoring points are attached. 
A reference position taking into account such a displacement can be provided to the leader robot at the equilibrium according to
\begin{equation}
\pREqIncRef{1} = \pREqInc{1} - \rotMatRRef{1}\displacement
\end{equation}
where $\rotMatRRef{1}$ is the leader robot rotational matrix at the equilibrium. This matrix is computed from the leader robot's equilibrium condition under the assumption that the thrust is aligned with the external forces
$$\rotMatRRef{1}\vE{3}=\frac{m_Rg\vE{3}+\cableForcesEqInc{1}}{||m_Rg\vE{3}+\cableForcesEqInc{1}||}:=\begin{bmatrix}A1\\A2\\A3\end{bmatrix}.$$
Hence, assuming the yaw is controlled to zero,  $\rotMatRRef{1}=\rotMat{_Y}(\pitch_{R1}^r)\rotMat{_X}(\roll_{R1}^r)$ with  ${\pitch_{R1}^r=\text{atan}{(A1/A3)}}$ and ${\roll_{R1}^r=\text{asin}{(-A2)}.}$ 
\bibliographystyle{IEEEtran}
\bibliography{Biblio/bibAlias,Biblio/bibMain,Biblio/bibNew,Biblio/bibAF,Biblio/bibCustom}
\end{document}